\documentclass[journal ]{new-aiaa}
\usepackage[utf8]{inputenc}
\usepackage{textcomp}

\usepackage{graphicx}
\usepackage{amsmath}

\usepackage{amsthm}

\usepackage[version=4]{mhchem}
\usepackage{siunitx}
\usepackage{longtable,tabularx}
\usepackage{algorithm}
\usepackage{algorithmic}
\usepackage{subfig}
\usepackage{graphicx}
\usepackage{dblfloatfix}
\usepackage{multirow}

\DeclareMathOperator{\argmin}{arg\,min}

\newtheorem{theorem}{\bf Theorem}
\newtheorem{lemma}{\bf Lemma}

\newtheorem{remark}{\bf Remark}

\newtheorem{assumption}{\bf Assumption}

\newcommand{\nnum}{\nonumber}

\newcommand{\RR}{\mathbb{R}}

\newcommand\oprocendsymbol{\hbox{$\blacksquare$}}
\newcommand\oprocend{\relax\ifmmode\else\unskip\hfill\fi\oprocendsymbol}

\title{Backup Plan Constrained Model Predictive Control with Guaranteed Stability}

\author{Ran Tao \footnote{Graudate student, Department of Mechanical Science and Engineering, rant3@illinois.edu. Corresponding author} }
\affil{University of Illinois Urbana-Champaign, Urbana, Illinois, 61801}
\author{Hunmin Kim \footnote{Assistant Professor, Department of Electrical and Computer Engineering, kim\_h@mercer.edu}}
\affil{Mercer University, Macon, Georgia, 31207}
\author{Hyung-Jin Yoon\footnote{Assistant Professor, Department of Mechanical Engineering, hyoon@tntech.edu}}
\affil{Tennessee Technological University, Cookeville, Tennessee, 38505}
\author{Wenbin Wan \footnote{Assistant Professor, Department of Mechanical Engineering, wwan@unm.edu}}
\affil{University of New Mexico, Albuquerque, New Mexico, 87131}
\author{Naira Hovakimyan \footnote{Professor, Department of Mechanical Science and Engineering, nhovakim@illinois.edu}, Lui Sha \footnote{Professor, Department of Computer Science, lrs@illinois.edu}}
\affil{University of Illinois Urbana-Champaign, Urbana, Illinois, 61801}
\author{Petros Voulgaris\footnote{Professor, Department of Mechanical Engineering, pvoulgaris@unr.edu}}
\affil{University of Nevada, Reno, Nevada, 89557}

\begin{document}

\maketitle
\begin{abstract}
This article proposes and evaluates a new safety concept called backup plan safety for path planning of autonomous vehicles under mission uncertainty using model predictive control (MPC). Backup plan safety is defined as the ability to complete an alternative mission when the primary mission is aborted. To include this new safety concept in control problems, we formulate a feasibility maximization problem aiming to maximize the feasibility of the primary and alternative missions. The feasibility maximization problem is based on multi-objective MPC, where each objective (cost function) is associated with a different mission and balanced by a weight vector. Furthermore, the feasibility maximization problem incorporates additional control input horizons toward the alternative missions on top of the control input horizon toward the primary mission, denoted as multi-horizon inputs, to evaluate the cost for each mission. We develop the backup plan constrained MPC algorithm, which designs the weight vector that ensures asymptotic stability of the closed-loop system, and generates the optimal control input by solving the feasibility maximization problem with computational efficiency. 
The performance of the proposed algorithm is validated through simulations of a UAV path planning problem.
\end{abstract}


\section{Introduction}\label{sec:intro}
The design and deployment of autonomous systems often rely on collision avoidance as the safety measure for path planning purposes. Nevertheless, the incident of Miracle on the Hudson (US Airways Flight 1549)~\citep{marra2009migratory} highlights the need for new safety definitions for automated systems that perform complicated tasks with mission uncertainties and safety priority. Due to both engines' failure from an unexpected bird strike, Captain Sullenberger had to abort the original flight plan. After checking the feasibility of safe landing points nearby, he successfully landed the plane on the Hudson river. However, a tragic plane crash and devastating human loss could have occurred if the plane could not complete the emergency landing due to the lack of feasible alternative landing locations. Motivated by this critical need, we introduce a new safety definition called {\em backup plan safety} for path planning of autonomous vehicles using MPC. Backup plan safety is defined as the ability to complete one of the alternative missions, e.g., landing at an alternative destination, in the case of the abortion of the primary mission, e.g., landing at the primary destination.

\begin{figure}[!h]
    \centering
    \includegraphics[width=.6\textwidth]{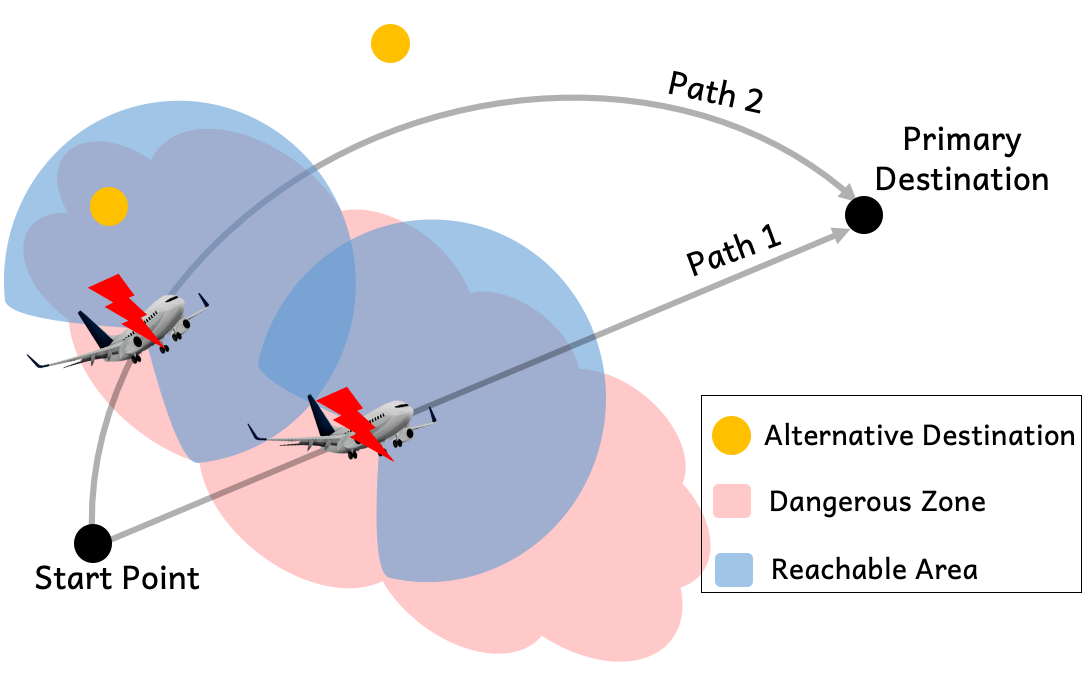}
    \caption{(Feasibility maximization scenario) The significance of the backup plan safety.}
    \label{fig:demo}
\end{figure}

In the path planning of airplanes, a safety concept similar to the backup plan safety, the 60-minute flight rule, was established in 1953 to limit the maximal distance to nearby airports along the planned trajectory for a twin-engine aircraft to guarantee a potential safe emergency landing. To fly outside of the 60-minute range, the aircraft needs to follow the extended-range twin-engine operational performance standards (ETOPS), which greatly enhance the safety of aircraft operations by certification of both the aircraft and the airline operator with more stringent FAA standards \citep{desantis2013engines}. Nevertheless, the 60-minute flight rule is limited to the path planning of aircraft. To generalize this rule for safety standards to a broad domain of applications, we propose the backup plan safety, which is beneficial to ensure safety during complex operations for all automated systems with long-horizon emergency response and mission uncertainties. The significance of the backup plan safety is illustrated in Fig.~\ref{fig:demo}, which demonstrates the situation of an aircraft traveling to its destination while passing through a hazardous zone, e.g., a storm. In Fig.\ref{fig:demo}, the red area represents the dangerous zone in which the plane may experience unexpected failures and not be able to arrive at the primary destination, and the blue area indicates the reachable area after the airplane is struck by lightning and has limited performance. Path 1 is a typical solution to the path planning problem as it has the shortest distance. However, along Path 1, there are no feasible means to safely land at an alternative airport when the plane confronts an emergency inside the dangerous zone and has limited performance. On the other hand, Path 2 increases the feasibility of a safe landing at other destinations. Thus, even when an emergency happens, following Path 2 allows the possibility for an airplane to finish a safe emergency landing at one of the two alternative destinations. The backup plan safety proposed in the current paper generalizes such a rule regarding the safety standard and application domains.

Model predictive control (MPC) is an optimal control methodology that uses an iterative approach to optimize control inputs for a finite time horizon while satisfying a set of constraints. MPC has been widely applied in the path and motion planning domain to ensure stability and safety \cite{paden2016survey,aggarwal2020path}.
Recent improvements in computing hardware and algorithmic innovations have made real-time MPC feasible, and there have been attempts to integrate machine learning with MPC techniques to achieve safe autonomy in a unified framework \citep{aswani2013provably,hewing2019cautious}. However, the safety goal in the literature of MPC applications is generally limited to collision avoidance. 

To address optimization problems with multiple and often conflicting performance criteria in the MPC framework, multi-objective MPC (MMPC) has been developed \citep{bemporad2009multiobjective}. In MMPC, different objectives are usually conflicting, and it is impossible to obtain a solution that optimizes all the objectives at the same time. Therefore, \citep{bemporad2009multiobjective} specifically focuses on finding the Pareto optimal solutions by balancing the cost functions of each objective with a weight vector. In addition to path planning, MMPC has demonstrated exceptional abilities in a variety of applications, such as power converter control \citep{hu2013multi,hu2018multi}, HVAC  (heating, ventilating, and air-conditioning) control \citep{ascione2017new}, and cruise control \citep{zhao2017real,li2010model}. For backup plan constrained control problems that involve multiple and conflicting missions, MMPC in \citep{bemporad2009multiobjective} provides a foundation for the feasibility maximization problem formulation. Nonetheless, incorporating backup plan safety requires modifications to the setup of MMPC. Specifically, as shown in Fig. \ref{fig:demo}, the cost for the plane to arrive at an alternative destination is well evaluated only when the trajectory (control inputs) toward that alternative destination is provided (e.g., mission feasibility). Thus, the single prediction horizon control included in both MPC and MMPC is insufficient for the backup plan safety, and additional inputs toward each alternative destination are required.
To ensure closed-loop stability in MPC, the usual approach is to consider the cost function of MPC as a candidate Lyapunov function and assume some properties of the cost function and the system dynamics to ensure a non-increasing cost function \cite{lazar2006stabilizing, bemporad2009multiobjective}. To achieve stability in MMPC, a fundamental assumption is having `aligned objectives' (i.e., all missions share the same destinations and all cost functions of each mission are minimized at the same state) \citep{bemporad2009multiobjective}. This assumption is conservative and unsuitable for backup plan safety where we consider missions with misaligned landing destinations, and each cost function should be minimized at the destination of the associated mission.

Stochastic Model Predictive Control (SMPC) leverages the probabilistic uncertainty model in an optimal control problem to balance the tradeoff between optimizing control objectives and satisfying chance constraints. Various approaches have been developed, including stochastic-tube \citep{cannon2010stochastic,cannon2012stochastic}, stochastic programming \citep{blackmore2010probabilistic}, and sampling-based approaches \citep{visintini2006monte,kantas2009sequential,williams2016aggressive,williams2017model, williams2018information},  to address SMPC. In particular, model predictive path integral control (MPPI) \citep{williams2016aggressive,williams2017model, williams2018information} is a sampling-based method  of SMPC that can handle nonlinear dynamical systems and uses Graphics Processing Units (GPUs) to enable efficient parallel computing. 

In this paper, we propose a new safety concept, backup plan safety, to enhance the safe operation of complex autonomous systems under mission uncertainty. This concept aims to maximize the feasibility of completing alternative missions, e.g., arriving at an alternative destination, during the operation of the primary mission, e.g., arriving at the primary destination, and we formulate a feasibility maximization problem to address the backup plan safety in control problems. Assigning each mission with an objective (cost function), we formulate the feasibility maximization problem based on multi-objective model predictive control (MPC). In particular, we use a weight vector to balance the multiple cost functions and incorporate additional (virtual) control input horizons toward the alternative missions on top of the control input horizon toward the primary mission, which we denote as multi-horizon inputs, to enable the evaluation of the costs for the alternative missions. With the inclusion of multi-cost functions and multi-horizon inputs, we are optimizing control inputs for every alternative mission. To this end, we develop the backup plan constrained MPC algorithm, which determines the weight vector that guarantees the asymptotic stability of the closed-loop system even when different missions have different destinations, which we denote as misaligned objectives, and generates the optimal control input by solving the feasibility maximization problem with multi-horizon inputs efficiently. To achieve closed-loop asymptotic stability, the designed weight vector ensures a non-increasing overall cost, i.e., a weighted sum of the cost functions of each mission, by making the cost decrease in the primary mission dominate the cost increase in the alternative missions. To reduce the computational complexity of solving the feasibility maximization problem with multi-horizon inputs, we propose the multi-horizon multi-objective model predictive path integral (3M) solver, which leverages a sampling-based scheme \citep{williams2016aggressive,williams2017model, williams2018information}.
Using simulations of an aerial vehicle model, we demonstrate this new safety concept and the performance of the proposed algorithm.

The remainder of the paper is organized as follows. Section~\ref{sec:MMPC} formulates the feasibility maximization problem based on multi-objective MPC with modifications to include multi-horizon inputs. Section \ref{sec:backup plan constrained mpc algorithm} presents the backup plan constrained MPC algorithm, and Section \ref{sec:proof} proves the asymptotic stability of the closed-loop system following the proposed algorithm. Section~\ref{sec:3M} introduces the 3M solver, which is leveraged by the backup plan constrained MPC algorithm to solve the feasibility maximization problem efficiently. Simulation results of UAV control problems are presented in Section~\ref{sec:sim}.

\section{Feasibility Maximization Problem Formulation}\label{sec:MMPC}
Consider the discrete-time dynamical system:
\begin{align}
    &\mathbf{x}_{k+1}= f(\mathbf{x}_k,\mathbf{u}_{k})\nnum\\
    & \text{s.t. } \mathbf{x}_k\in \mathbb{X},~~ \mathbf{u}_{k}\in \mathbb{U},
    \label{eq:sysmodel}
\end{align}
where $f:\RR^{n_x}\times\RR^{n_u}\rightarrow \RR^{n_x}$ is a continuous nonlinear function, $\mathbf{x}_k \in \RR^{n_x}$ is the system state, $\mathbf{u}_{k} \in \RR^{n_u}$ is the control input at time $k \geq 0$, and $\mathbb{X}$ and $\mathbb{U}$ are two compact sets of the constraints. It is assumed that one primary mission and $m$ alternative missions are provided. Without loss of generality, in this paper, we specifically focus on missions as arriving at destinations, and we say mission $i$ is completed at time $t$, if
\begin{align}
   t=&\argmin_k k\nnum\\
   &{\rm \ s.t. \ } d(\mathbf{x}_k,\mathbf{p}^i)=0 ,
   \label{eq:MissionCompletetion}
\end{align}
where $d:\RR^{n_x}\times\RR^{n_x}\rightarrow \RR_{\geq0}$ is the distance metric between the system state $\mathbf{x}_k$ and the destination of mission $i$, which is denoted by $\mathbf{p}^i \in \RR^{n_x}$. Without loss of generality, we further assume that $\mathbf{p}^0=0$. The control objective for the system~\eqref{eq:sysmodel} is to arrive at the primary destination, while the trajectory toward the primary destination is safe in terms of backup plan safety. 

As we are simultaneously maximizing the feasibility of the primary and alternative destinations, based on MMPC \cite{bemporad2009multiobjective}, we formulate the control problem of the system~\eqref{eq:sysmodel} as the following optimization problem, which we call feasibility maximization:
\begin{align}
    &\min_{\mathbf{U}} \mathbf{J}(\mathbf{x},\mathbf{U})\nnum\\
    &\text{s.t.} \nnum \\
    &\mathbf{x}_{k+1}= f(\mathbf{x}_k,\mathbf{u}_{k}), \mathbf{x}_0=\mathbf{x}\nnum\\
    &\mathbf{x}_k\in \mathbb{X}, k = 1,\dots,N \nnum \\
    &\mathbf{u}_{k}\in \mathbb{U}, k = 0,\dots,N-1,
    \label{eq:mmpc0}
\end{align}
where $\mathbf{x}$ is the current state of the system, $N$ is the prediction horizon, $\mathbf{U} \in \RR^{(N+\frac{N(N-1)m}{2})n_u}$ is the multi-horizon control input sequence, and $\mathbf{J}: \RR^{n_x} \times \RR^{(N+\frac{N(N-1)m}{2})n_u} \rightarrow \RR^{m+1}$ is an $m+1$ dimensional vector cost function with the first element representing the cost function of the primary destinations and the $j^{th}$ element associated with $(j-1)^{th}$ alternative destinations. 
Similar to standard MPC, at every time step, we execute the first control input of $\mathbf{U}$, shift the prediction horizon forward, and optimize the cost function again. 
Unlike the standard MPC or MMPC, the control input sequence $\mathbf{U}$ from the feasibility maximization problem~\eqref{eq:mmpc0} includes an input horizon toward the primary destination and additional (virtual) input horizons toward alternative destinations, which helps to evaluate the cost functions toward both the primary and alternative destinations.
Section \ref{sec:input} discusses the dimension and elements of the multi-horizon control input sequence $\mathbf{U}$, and Section \ref{sec:cost} discusses the design of the vector cost function $\mathbf{J}$.

Minimizing the first element in $\mathbf{J}$ maximizes the feasibility of the primary destination. Minimizing the $j^{th}$ element in $\mathbf{J}$ maximizes the feasibility of the $(j-1)^{th}$ alternative destination.
By solving~\eqref{eq:mmpc0}, we can obtain a control sequence $\mathbf{U}$ that maximizes the feasibility of all destinations simultaneously.
However, as in a multi-objective optimization problem, the feasibility maximization problem~\eqref{eq:mmpc0} includes conflicting cost functions in $\mathbf{J}(\mathbf{x},\mathbf{U})$, and, in this case, no solution minimizes all the cost functions simultaneously. Thus, Pareto optimality plays a key role in defining optimality for multi-objective optimization \citep{marler2004survey,bemporad2009multiobjective}.
Similar to MMPC in \citep{bemporad2009multiobjective}, we assign parameterized weights to the cost functions so that we limit our focus to a particular subset of Pareto optimal solutions. With the weight assignment, we reformulate the feasibility maximization problem ~\eqref{eq:mmpc0} as
\begin{align}
 &\min_{\mathbf{U}} {\boldsymbol{\alpha}}^\top\mathbf{J}(\mathbf{x},\mathbf{U})\nnum\\
 &\text{s.t.}\nnum\\
 &\mathbf{x}_{k+1}= f(\mathbf{x}_k,\mathbf{u}_{k}), \mathbf{x}_0=\mathbf{x}\nnum \\
 &\mathbf{x}_k\in \mathbb{X}, k = 1,\dots,N \nnum \\
 &\mathbf{u}_{k}\in \mathbb{U}, k = 0,\dots,N-1, 
 \label{eq:mmpc}
 \end{align}
where the weight vector $\boldsymbol{\alpha} = [\alpha^0,\cdots,\alpha^m]^\top \in \RR^{m+1}$
satisfies $\alpha^i \in [0,1] \subset \RR$ and $\sum_{i=0}^{m}\alpha^i=1$.
If $\alpha^i>0$ $\forall i$, the solution to~\eqref{eq:mmpc} is a Pareto optimal solution of~\eqref{eq:mmpc0}, while this may not hold when there exist some $i$ such that $\alpha^i=0$ \citep{boyd2004convex}.

\subsection{Multi-horizon Control Input Sequence and State Trajectory}\label{sec:input}
The multi-horizon control input sequence
\begin{align}
  \mathbf{U} = [(\mathbf{U}^0)^\top,(\mathbf{U}^1)^\top,\cdots,(\mathbf{U}^m)^\top]^\top
  \label{eq:BigU}
\end{align}
consists of future control inputs toward the primary destination $\mathbf{U}^0 = [\mathbf{u}_0^\top,\mathbf{u}_1^\top,\cdots,\mathbf{u}_{N-2}^\top,\mathbf{u}_{N-1}^\top]^\top$,
which is the input horizon used in the standard MPC, and additional (virtual) input horizons toward the alternative destinations $\mathbf{U}^i$ for $i=1,\cdots,m$, where we have $m$ alternative destinations in total.
Based on the definition of backup plan safety, two important properties should be followed when defining $\mathbf{U}^i$ toward the alternative destination $i$: It is unknown when the system will abort the primary destination and the control horizon toward the alternative destinations should be $N$.
According to the first property, $\mathbf{U}^i$ should consider the scenario of primary destination abortion at any time along the trajectory to the primary destination, and thus we construct $\mathbf{U}^i$ as:
\begin{align}
    \mathbf{U}^i = [(\mathbf{U}_0^i)^\top,(\mathbf{U}_1^i)^\top,\cdots,(\mathbf{U}_{N-3}^i)^\top,(\mathbf{U}_{N-2}^i)^\top]^\top,
    \label{eq:MediumU}
\end{align}
where $\mathbf{U}_\mathbf{p}^i$ consists of control inputs toward the $i^{th}$ alternative destination when the primary destination is aborted after $p+1$ time steps, i.e., $\mathbf{u}_p$ from $\mathbf{U}^0$ has been executed. Following this design, $\mathbf{u}_0,\cdots,\mathbf{u}_p$ from $\mathbf{U}^0$ are the first $p+1$ elements of $\mathbf{U}_\mathbf{p}^i$. Because $\mathbf{U}_\mathbf{p}^i$ should include $N$ elements to satisfy the second property, an additional $N-(p+1)$ inputs toward the alternative destination $i$ should be included in $\mathbf{U}_\mathbf{p}^i$.
Then, we have
\begin{align*}
    \mathbf{U}_\mathbf{p}^i = [\mathbf{u}_0^\top,\cdots,\mathbf{u}_\mathbf{p}^\top,(\mathbf{u}_{p,p+1}^i)^\top,\cdots,(\mathbf{u}_{p,N-1}^i)^\top]^\top,
\end{align*}
where $\mathbf{u}_{p,q}^i$ is the $(q+1)^{th}$ control input when the primary destination is aborted after executing $\mathbf{u}_p$ and we choose the alternative destination $i$ for an emergency landing.
\begin{figure}[t]
 \centering
  \includegraphics[width=0.7\linewidth]{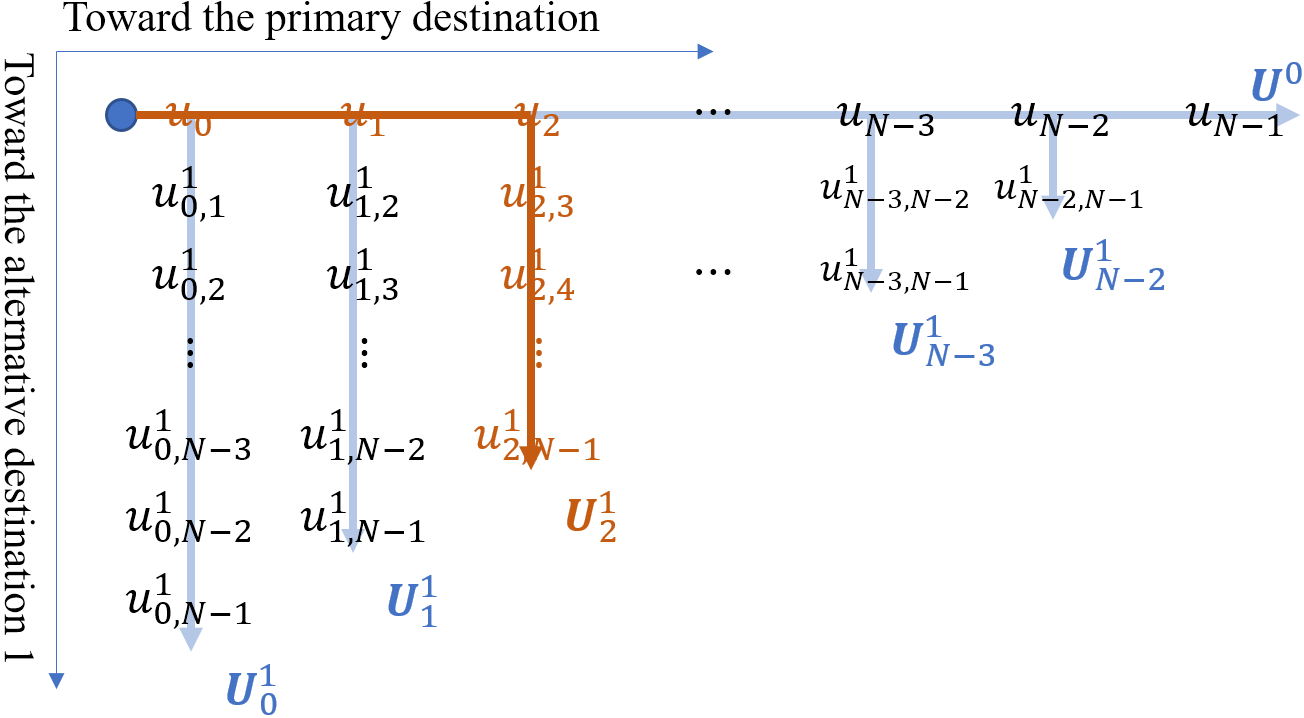}
  \caption{The elements of multi-horizon control input sequence $\mathbf{U}$.}
\label{fig:inputs}
\end{figure}
A visualization of the multi-horizon control input sequence $\mathbf{U}$ is presented in Fig.~\ref{fig:inputs} for the case of only one alternative destination. The figure depicts the relation between $\mathbf{U}$, $\mathbf{U}^0$, and $\mathbf{U}_\mathbf{p}^1$. To construct $\mathbf{U}^0$ or $\mathbf{U}_\mathbf{p}^1$, we start at the blue dot at the left upper corner, move right until the primary destination is aborted, and move downside afterward as alternative destination $1$ is chosen. Considering the multi-horizon inputs and the feasibility maximization problem \eqref{eq:mmpc}, we aim to optimize the control inputs toward each alternative destination at every time step.

The multi-horizon state trajectory $\mathbf{X}= [(\mathbf{X}^0)^\top,(\mathbf{X}^1)^\top,\cdots,(\mathbf{X}^m)^\top]^\top$ is constructed by simulating the multi-horizon control input sequence $\mathbf{U}$ on the system~\eqref{eq:sysmodel}, where $\mathbf{X}^0$ and $\mathbf{X}^i=[(\mathbf{X}_0^i)^\top,\cdots,(\mathbf{X}_{N-2}^i)^\top]^\top$ consist of the simulated states based on inputs $\mathbf{U}^0$ and $\mathbf{U}^i$ respectively for $i=1,\cdots,m$.
Similarly, $\mathbf{X}_\mathbf{p}^i = [\mathbf{x}_0^\top,\cdots,\mathbf{x}_\mathbf{p}^\top,\mathbf{x}_{p+1}^\top,\cdots,(\mathbf{x}_{p+1,N}^i)^\top]^\top$ is the vector of simulated states based on the input $\mathbf{U}_\mathbf{p}^i$, for $i=1,\cdots,m$ and $p=0,\cdots,N-2$.

Since $\mathbf{U}_\mathbf{p}^i$ and $\mathbf{U}^0$ share the first $p+1$ elements, $\mathbf{X}_\mathbf{p}^i$ and $\mathbf{X}^0$ share the first $p+2$ elements, which are $\mathbf{x}_0,\mathbf{x}_1,\cdots,\mathbf{x}_p,\mathbf{x}_{p+1}$. Therefore, the multi-horizon control input sequence $\mathbf{U}$ consists of the $N$ number of inputs toward the primary destination in $\mathbf{U}^0$, and $N-1,N-2,\cdots,1$ additional number of inputs toward the alternative destination $i$ in $\mathbf{U}_0^i,\mathbf{U}_1^i,\cdots,\mathbf{U}_{N-2}^i$. As a result, a total number of $N+\frac{N(N-1)}{2}m$ independent elements are in $\mathbf{U}$, and $\mathbf{X}$ includes a total number of $N+1+\frac{N(N-1)}{2}m$ independent elements. To address the induced computational complexity from the multi-horizon control input sequence, the proposed backup plan constrained MPC algorithm leverages the sampling-based 3M solver described in Section \ref{sec:3M} to solve the feasibility maximization problem efficiently. 

\subsection{Multi-objective Vector Cost Function}\label{sec:cost}
Given a state $\mathbf{x}$ and a multi-horizon control input sequence $\mathbf{U}$, $\mathbf{J}(\mathbf{x},\mathbf{U})$ is an $m+1$ dimensional vector function defined as:
\begin{align*}
    \mathbf{J}(\mathbf{x},\mathbf{U}) = [\mathbf{J}^0(\mathbf{x},\mathbf{U}^0),\mathbf{J}^1(\mathbf{x},\mathbf{U}^1),\cdots,\mathbf{J}^m(\mathbf{x},\mathbf{U}^m)]^\top,
\end{align*}
where $\mathbf{J}^0(\mathbf{x},\mathbf{U}^0)$ represents the cost function of the primary destination as $\mathbf{U}^0$ consists of the control inputs toward the primary destination, and $\mathbf{J}^i(\mathbf{x},\mathbf{U}^i)$ represents the cost function of the alternative destination $i$ as $\mathbf{U}^i$ consists of control inputs toward the alternative destination $i$. Since it is unknown when the system will abort the primary destination, we assume that diverting to the alternative destination is equally likely at each time step, and we design the elements of $\mathbf{J}(\mathbf{x},\mathbf{U})$ as the average cost over state-input trajectories toward the corresponding destination:
\begin{align*}
    &\mathbf{J}^0(\mathbf{x},\mathbf{U}^0) =J^0(\mathbf{x},\mathbf{U}^0),\nnum\\
    &\mathbf{J}^i(\mathbf{x},\mathbf{U}^i) =
    \frac{1}{N-1}\sum_{p=0}^{N-2}J^i(\mathbf{x},\mathbf{U}_\mathbf{p}^i)
\end{align*}
for $i=1,\cdots,m$. Let $\mathbf{X}= [(\mathbf{X}^0)^\top,(\mathbf{X}^1)^\top,\cdots,(\mathbf{X}^m)^\top]^\top$ be the corresponding multi-horizon state trajectory by simulating $\mathbf{U}$ on system \eqref{eq:sysmodel} at state $\mathbf{x}$. The function $J^i$ is a standard cost function for the destination $i$:
\begin{align}
    J^i(\mathbf{x},\mathbf{U}^i_p) {=} \sum_{k=1}^{N} L^i(\mathbf{X}^i_p[k],\mathbf{U}^i_p[k]) + F^i(\mathbf{X}^i_p[N+1])
    \label{eq:costSmallJ}
\end{align}
that consists of the running cost $L^i$ and the terminal cost $F^i$, where $\mathbf{X}^i_p[k] \in \RR^{n_x}$ and $\mathbf{U}^i_p[k] \in \RR^{n_u}$ are the $k^{th}$ state and input of $\mathbf{X}^i_p$ and $\mathbf{U}^i_p$, respectively.
With the multi-objective cost functions and multi-horizon inputs, we optimize control inputs for every alternative mission starting at every time step through the feasibility maximization problem \eqref{eq:mmpc}.

\subsection{Weight Vector}\label{sec:weight}
As the weight vector $\boldsymbol{\alpha}$ in the feasibility maximization problem \eqref{eq:mmpc} determines the cost function for optimization, it affects the performance and stability of the closed-loop system significantly. In this paper, we will design the weight vector as a function depending on the state $\mathbf{x}_k$, which we denote as $\boldsymbol{\alpha}_*(\mathbf{x}_k)$, and we denote the solution to the feasibility maximization \eqref{eq:mmpc} for $\mathbf{x}=\mathbf{x}_k$ and $\boldsymbol{\alpha}=\boldsymbol{\alpha}_*(\mathbf{x}_k)$ by the optimal control input sequence $\mathbf{U}_*(\mathbf{x}_k)$.

\section{Backup Plan Constrained MPC Algorithm}
\label{sec:backup plan constrained mpc algorithm}
The current section proposes the backup plan constrained MPC algorithm which at time step $k$ with state $\mathbf{x}_k$ determines $\boldsymbol{\alpha}_*(\mathbf{x}_k)$, then leverages the 3M solver to solve the feasibility maximization problem \eqref{eq:mmpc} for $\mathbf{x}=\mathbf{x}_k$ and $\boldsymbol{\alpha}=\boldsymbol{\alpha}_*(\mathbf{x}_k)$, and finally determines the optimal control input $\mathbf{u}^*_k$. 
The details of the 3M solver will be explained in Section \ref{sec:3M}, and the proposed backup plan constrained MPC algorithm is summarized in Algorithm \ref{al:backup}.
\begin{algorithm}[H]
\caption{Backup Plan Constrained MPC Algorithm}\label{al:backup}
\textbf{Input to the Backup Plan Constrained MPC algorithm:}\vspace{-2mm}\\
\vspace{-2mm}
$\mathbf{x}_k$: State of the system at the current time step; \\\vspace{-2mm}
$\boldsymbol{\alpha}_*(\mathbf{x}_{k-1})$: Weight vector applied at the previous time step;\\\vspace{-2mm}
$\mathbf{U}_*(\mathbf{x}_{k-1})$: Optimal control input sequence at the previous time step;\\\vspace{-2mm}
\textbf{Tuning parameters:}
Positive constants $\delta$, $\gamma^i$ ($i=0,1,\cdots,m$), and $\mu$, matrix $K$, control input $\hat{\mathbf{u}}$

\begin{algorithmic}[1]
\STATE Compute $\mathbf{U}_s(\mathbf{x}_k)$ in \eqref{eq:u_s} based on $\mathbf{U}_*(\mathbf{x}_{k-1})$ with parameters $K$ and $\hat{\mathbf{u}}$;
\vspace{-2mm}
\IF{$\mathbf{x}_k\in B_\delta$} \vspace{-2mm}
    \STATE Set $\boldsymbol{\alpha}_*(\mathbf{x}_k)=[1,0,0,\dots,0]^\top$;\vspace{-2mm}
    \STATE Compute $\mathbf{U}_*(\mathbf{x}_k)$ by solving \eqref{eq:mmpc} for $\mathbf{x}=\mathbf{x}_k$, $\boldsymbol{\alpha}=\boldsymbol{\alpha}_*(\mathbf{x}_k)$ with 3M solver, $\mathbf{U}_*(\mathbf{x}_k)=\text{3M}(\mathbf{x}_k,\boldsymbol{\alpha}_*(\mathbf{x}_k),\mathbf{U}_s(\mathbf{x}_{k}))$; \vspace{-2mm}
\ELSIF{$\mathbf{x}_k\notin B_\delta$}
  \vspace{-2mm}
      \STATE Compute the baseline weight vector $\boldsymbol{\alpha}_b(\mathbf{x}_k)$ in \eqref{eq:alpha_b} using parameters $\delta$, $\gamma^i$ ($i=0,1,\cdots,m$), $\mu$;
    \vspace{-2mm}
    \STATE Based on $\mathbf{U}_s(\mathbf{x}_k)$, $\boldsymbol{\alpha}_b(\mathbf{x}_k)$, and $\boldsymbol{\alpha}_*(\mathbf{x}_{k-1})$, compute the transitional weight vector $\boldsymbol{\alpha}_t(\mathbf{x}_k)$ in \eqref{eq:alpha_t};
    \vspace{-2mm}
    \STATE Compute $\mathbf{U}(\mathbf{x}_k,\boldsymbol{\alpha}_t(\mathbf{x}_k))$ by solving \eqref{eq:mmpc} for $\mathbf{x}=\mathbf{x}_k$, $\alpha=\boldsymbol{\alpha}_t(\mathbf{x}_k)$ with 3M solver,\\\vspace{-2mm} $\mathbf{U}(\mathbf{x}_k,\boldsymbol{\alpha}_t(\mathbf{x}_k))=\text{3M}(\mathbf{x}_k,\boldsymbol{\alpha}_t(\mathbf{x}_k),\mathbf{U}_s(\mathbf{x}_{k}))$;
    \vspace{-2mm}
    \STATE Compute $\mathbf{x}_{k,f}$ by simulating $\mathbf{U}^0(\mathbf{x}_k,\boldsymbol{\alpha}_t(\mathbf{x}_k))$ on system \eqref{eq:sysmodel} at $\mathbf{x}_k$;
    \vspace{-2mm}
    \IF{$\mathbf{x}_{k,f}\in B_\delta$} \vspace{-2mm}
        \STATE Set $\boldsymbol{\alpha}_*(\mathbf{x}_k)=[1,0,0,\dots,0]^\top$;
        \vspace{-2mm}
        \STATE Compute $\mathbf{U}_*(\mathbf{x}_k)=\text{3M}(\mathbf{x}_k,\boldsymbol{\alpha}_*(\mathbf{x}_k),\mathbf{U}_s(\mathbf{x}_{k}))$;
        \vspace{-2mm}
        \ELSIF{$\mathbf{x}_{k,f}\notin B_\delta$} \vspace{-2mm}
            \STATE Set $\boldsymbol{\alpha}_*(\mathbf{x}_k)=\boldsymbol{\alpha}_t(\mathbf{x}_k)$ ;
            \vspace{-2mm}
            \STATE Set $\mathbf{U}_*(\mathbf{x}_k)=\mathbf{U}(\mathbf{x}_k,\boldsymbol{\alpha}_t(\mathbf{x}_k))$, where $\mathbf{U}(\mathbf{x}_k,\boldsymbol{\alpha}_t(\mathbf{x}_k))$ is achieved in line 8;
            \vspace{-2mm}
     \ENDIF \vspace{-2mm}
\ENDIF \vspace{-2mm}
\IF{$\boldsymbol{\alpha}_*(\mathbf{x}_{k-1})=[1,0,0,\dots,0]^\top$} \vspace{-2mm}
    \STATE Set $\boldsymbol{\alpha}_*(\mathbf{x}_k)=[1,0,0,\dots,0]^\top$; \vspace{-2mm}
    \STATE Compute $\mathbf{U}_*(\mathbf{x}_k)=\text{3M}(\mathbf{x}_k,\boldsymbol{\alpha}_*(\mathbf{x}_k),\mathbf{U}_s(\mathbf{x}_{k}))$; \vspace{-2mm}
\ENDIF \vspace{-2mm}
\STATE Set $\mathbf{u}^*_k$ as the first element of $\mathbf{U}_*(\mathbf{x}_k)$ and apply $\mathbf{u}^*_k$ to the system.
\end{algorithmic}
\end{algorithm}

Before presenting the details, a summary of the proposed backup plan constrained MPC algorithm proceeds as follows. While the current state $\mathbf{x}_k$ and the final state toward the primary destination $\mathbf{x}_{k,f}$, which will be defined later, are outside an open ball $B_\delta$ centered at the origin with radius $\delta$ (line 5 and line 13), we specifically design $\boldsymbol{\alpha}_*(\mathbf{x}_k)$ (line 14) to ensure a non-increasing overall cost, i.e., the weighted sum of the cost functions of each destination. During this time, the backup plan constrained MPC algorithm considers arriving at alternative destinations during the operation toward the primary destination. Once either $\mathbf{x}_k$ or $\mathbf{x}_{k,f}$ enters $B_\delta$ (line 2 and line 10), we set $\boldsymbol{\alpha}_*(\mathbf{x}_k)=[1,0,0,\dots,0]^\top$ (line 3 and line 11) for the current and all future time steps (line 18-19), and the backup plan constrained MPC algorithm considers arriving at the primary destination only thereafter. 

The details of the algorithm proceed as follows.
At the current time step $k$, we know the values of $\boldsymbol{\alpha}_*(\mathbf{x}_{k-1})$ and $\mathbf{U}_*(\mathbf{x}_{k-1})$, which are the designed weight vector and the optimal control input sequence at the previous time step. With $\mathbf{U}_*(\mathbf{x}_{k-1})$, we first calculate a control input sequence
\begin{align}
\label{eq:u_s}
    \mathbf{U}_s(\mathbf{x}_{k}) = [\mathbf{U}^0_s(\mathbf{x}_{k})^\top, \mathbf{U}^1_s(\mathbf{x}_{k})^\top,\dots.,\mathbf{U}^m_s(\mathbf{x}_{k})^\top]^\top,
\end{align}
which is constructed from $\mathbf{U}_*(\mathbf{x}_{k-1})=[\mathbf{U}^0_*(\mathbf{x}_{k-1})^\top, \mathbf{U}^1_*(\mathbf{x}_{k-1})^\top,\dots.,\mathbf{U}^m_*(\mathbf{x}_{k-1})^\top]^\top$ with parameters $K$ and $\hat{\mathbf{u}}$ (line 1). To be more specific, for $i=0,1,\dots,m$, $\mathbf{U}^i_s(\mathbf{x}_{k})$ is formed by removing the first element of $\mathbf{U}^i_*(\mathbf{x}_{k-1})$ and adding an extra element at the end of $\mathbf{U}^i_*(\mathbf{x}_{k-1})$. For $\mathbf{U}^0_s(\mathbf{x}_{k})$, the extra element is $K\mathbf{x}_{k-1,f}$, where $K$ is the tuning parameter and $\mathbf{x}_{k-1,f}$ is the resulting state toward the primary destination from forward integrating with the control inputs specified in $\mathbf{U}_*^0(\mathbf{x}_{k-1})$ on system \eqref{eq:sysmodel} at $\mathbf{x}_{k-1}$. For $\mathbf{U}^i_s(\mathbf{x}_{k})$ s.t. $i \neq0$, the extra element we add is the tuning parameter $\hat{\mathbf{u}}$. For example, if $\mathbf{U}_*^0(\mathbf{x}_{k-1}) = [\mathbf{u}_0^\top,\mathbf{u}_1^\top,\dots,\mathbf{u}_{N-1}^\top]^\top$, then $\mathbf{U}_s^0(\mathbf{x}_{k}) = [\mathbf{u}_1^\top,\dots,\mathbf{u}_{N-1}^\top,(K\mathbf{x}_{k-1,f})^\top]^\top$. The calculated $\mathbf{U}_s(\mathbf{x}_k)$ will be used in the 3M solver to solve the feasibility maximization problem \eqref{eq:mmpc}. 

Next, we check if $\mathbf{x}_k$ is within $B_\delta$, which is an open ball centered at the origin with radius $\delta$. If $\mathbf{x}_k \in B_\delta$ (line 2), the current state is close to the origin, and we set $\boldsymbol{\alpha}_*(\mathbf{x}_k)=[1,0,0,\dots,0]^\top$ (line 3). We then achieve the optimal control input sequence $\mathbf{U}_*(\mathbf{x}_k)$ by solving the feasibility maximization problem \eqref{eq:mmpc} for $\mathbf{x}=\mathbf{x}_k$, $\boldsymbol{\alpha}=\boldsymbol{\alpha}_*(\mathbf{x}_k)$ with the 3M solver, $\mathbf{U}_*(\mathbf{x}_k)=\text{3M}(\mathbf{x}_k,\boldsymbol{\alpha}_*(\mathbf{x}_k),\mathbf{U}_s(\mathbf{x}_k))$ (line 4). The details of the 3M solver will be explained in Section \ref{sec:3M}.
On the other hand, if $\mathbf{x}_k\notin B_\delta$ (line 5), we first compute a baseline weight vector $\boldsymbol{\alpha}_b(\mathbf{x})=[\alpha_b^0(\mathbf{x}),\alpha_b^1(x),\dots,\alpha_b^m(x)]^\top$ using tuning parameters $\gamma^i$ ($i=0,1,\cdots,m$) and $\mu$ (line 6):
\begin{equation}\label{eq:alpha_b}
\boldsymbol{\alpha}_b(\mathbf{x}):
\left\{
  \begin{array}{@{}ll@{}}
    \alpha_b^i(\mathbf{x})=\gamma^i\left(\frac{\|\mathbf{x}\|}{\max(\mu,\|\mathbf{x}-\mathbf{p}^i\|)}\right) & \text{for $i=1,2,\dots,m$}\\
    \alpha_b^0(\mathbf{x})=1- \sum_{i=1}^{m}\alpha_b^i(\mathbf{x}).
  \end{array}\right.
\end{equation}
Based on the previously calculated $\mathbf{U}_s(\mathbf{x}_k)$ in line 1, $\boldsymbol{\alpha}_b(\mathbf{x}_k)$ in line 6, and $\boldsymbol{\alpha}_*(\mathbf{x}_{k-1})$, we calculate a transitional weight vector $\boldsymbol{\alpha}_t(\mathbf{x}_k)$ defined as (line 7):
\begin{equation}
    \label{eq:alpha_t}
    \boldsymbol{\alpha}_t(\mathbf{x}_k) =
    \left\{
      \begin{array}{@{}ll@{}}
        \boldsymbol{\alpha}_b(\mathbf{x}_k) &\text{if $\boldsymbol{\alpha}_b(\mathbf{x}_{k})^\top \mathbf{J}(\mathbf{x}_{k},\mathbf{U}_s(\mathbf{x}_{k}))\le \boldsymbol{\alpha}_*(\mathbf{x}_{k-1})^\top \mathbf{J}(\mathbf{x}_{k},\mathbf{U}_s(\mathbf{x}_k))$}, \\
        \boldsymbol{\alpha}_*(\mathbf{x}_{k-1}) & \text{otherwise}.
      \end{array}\right.
    \end{equation}
The design of $\boldsymbol{\alpha}_t(\mathbf{x}_k)$ in \eqref{eq:alpha_t} will be used in Section \ref{sec:proof} to prove the asymptotic stability of the closed-loop system.
Next, we compute $\mathbf{U}(\mathbf{x}_k,\boldsymbol{\alpha}_t(\mathbf{x}_k))= [\mathbf{U}^0(\mathbf{x}_k,\boldsymbol{\alpha}_t(\mathbf{x}_k))^\top, \mathbf{U}^1(\mathbf{x}_k,\boldsymbol{\alpha}_t(\mathbf{x}_k)),\dots,\mathbf{U}^m(\mathbf{x}_k,\boldsymbol{\alpha}_t(\mathbf{x}_k))]^\top$ by solving \eqref{eq:mmpc} for $\mathbf{x}=\mathbf{x}_k$, $\boldsymbol{\alpha}=\boldsymbol{\alpha}_t(\mathbf{x}_k)$ with the 3M solver, $\mathbf{U}(\mathbf{x}_k,\boldsymbol{\alpha}_t(\mathbf{x}_k))=\text{3M}(\mathbf{x}_k,\boldsymbol{\alpha}_t(\mathbf{x}_k),\mathbf{U}_s(\mathbf{x}_k))$ (line 8). 
When $\mathbf{x}_k\notin B_\delta$ (line 5-17), an additional check on $\mathbf{x}_{k,f}$, which is the resulting state by simulating $\mathbf{U}^0(\mathbf{x}_k,\boldsymbol{\alpha}_t(\mathbf{x}_k))$ from $\mathbf{U}(\mathbf{x}_k,\boldsymbol{\alpha}_t(\mathbf{x}_k))$ in line 8 on system \eqref{eq:sysmodel} at $\mathbf{x}_k$ (line 9), is required to determine $\boldsymbol{\alpha}_*(\mathbf{x}_k)$ to ensure a non-increasing overall cost. If $\mathbf{x}_{k,f}\in B_\delta$ (line 10), we set $\boldsymbol{\alpha}_*(\mathbf{x}_k) = [1,0,0,\dots,0]^\top$ (line 11) and compute the optimal control input sequence $\mathbf{U}_*(\mathbf{x}_k)$ by solving \eqref{eq:mmpc} for $\mathbf{x}=\mathbf{x}_k$, $\boldsymbol{\alpha}=\boldsymbol{\alpha}_*(\mathbf{x}_k)$ with the 3M solver,  $\mathbf{U}_*(\mathbf{x}_k)=\text{3M}(\mathbf{x}_k,\boldsymbol{\alpha}_*(\mathbf{x}_k),\mathbf{U}_s(\mathbf{x}_k))$ (line 12). 
Else if $\mathbf{x}_{k,f}\notin B_\delta$ (line 13), we set $\boldsymbol{\alpha}_*(\mathbf{x}_k)=\boldsymbol{\alpha}_t(\mathbf{x}_k)$ (line 14) and it follows that $\mathbf{U}_*(\mathbf{x}_k)=\mathbf{U}(\mathbf{x}_k,\boldsymbol{\alpha}_t(\mathbf{x}_k))$ with $\mathbf{U}(\mathbf{x}_k,\boldsymbol{\alpha}_t(\mathbf{x}_k))$ calculated in line 8 (line 15). 
Another key design of the weight vector is that once we have $\boldsymbol{\alpha}_*(\mathbf{x}_k)=[1,0,0,\dots,0]^\top$ for the current time step, we set $\boldsymbol{\alpha}_*(\mathbf{x}_k)$ as $[1,0,0,\dots,0]^\top$ for all future time steps regardless of whether $\mathbf{x}_k\notin B_\delta$ and $\mathbf{x}_{k,f}\notin B_\delta$ for some future time step (line 18-19). In other words, once line 2 or line 10 is satisfied for the first time following the backup plan constrained MPC algorithm, we set $\boldsymbol{\alpha}_*(x)=[1,0,0,\dots,0]^\top \; \forall \mathbf{x}\in\mathbb{X}$ thereafter.
In the end, we set the optimal control input $\mathbf{u}^*_k$ at the current time step as the first element of $\mathbf{U}_*(\mathbf{x}_k)$, and apply $\mathbf{u}^*_k$ to the system (line 22).

\section{Asymptotic Stability of the Closed-loop System}
\label{sec:proof}
In this section, we first list and explain the assumptions we have for the backup plan constrained control problems. Then, we prove that the existence of the tuning parameters of the backup plan constrained MPC algorithm satisfying two introduced properties is guaranteed. Next, we show that with the tuning parameters satisfying two introduced properties, a desired property related to the overall cost can be achieved. Finally, by leveraging the previous results, we prove that the system \eqref{eq:sysmodel} in the closed loop, following the backup plan constrained MPC algorithm proposed in Algorithm \ref{al:backup}, is asymptotically stable.

In order to achieve asymptotic stability for MMPC, \cite{bemporad2009multiobjective} includes a conservative assumption that all objectives are aligned, i.e., $\mathbf{p}^i=\mathbf{0} \; \forall i$ where $\mathbf{p}^i$ is the destination point of mission $i$, and $L^i(\mathbf{0},\mathbf{0})=F^i(\mathbf{0})=0\; \forall i$. For backup plan constrained control problems with misaligned objectives (different missions have different destinations), we include a general assumption below that the cost function of destination $i$ is minimized at $\mathbf{p}^i$ as follows.
\begin{assumption}
\label{assumption:1}
The primary destination is at $\mathbf{p}^0=\mathbf{0}$, which is the equilibrium point of the system. For $i=0,1,\dots,m$, the running cost $L^i(\mathbf{x},\mathbf{u})$ and terminal cost $F^i(\mathbf{x})$ are continuous with respect to x with $L^i(\mathbf{p}^i,\mathbf{0})=0$ and $F^i(\mathbf{p}^i)=0$, where $\mathbf{p}^i$ is the destination $i$. There exist class $\mathcal{K}$ functions, i.e., $\mathcal{K}$-functions, $\sigma_1 \text{ and } \sigma_2$, s.t. $\forall \mathbf{x} \in \mathbb{X}$ and $\forall \mathbf{u} \in \mathbb{U}$, $L^i(\mathbf{x},\mathbf{u})\ge\sigma_1(\lVert \mathbf{x}-\mathbf{p}^i \rVert)$, $F^i(\mathbf{x}) \ge\sigma_1(\lVert \mathbf{x}-\mathbf{p}^i \rVert)$, and $F^i(\mathbf{x}) \le\sigma_2(\lVert \mathbf{x}-\mathbf{p}^i \rVert)$. 
\end{assumption}

The second assumption below is related to the cost of the primary destination.
\begin{assumption}
\label{assumption 2}
For any positive parameter $\delta>0$ s.t. $B_\delta \subseteq \mathbb{X}$, where $B_\delta$ is the open ball centered at the origin with radius $\delta$, there exists a linear feedback $\mathbf{u}=K\mathbf{x}$, s.t.
\begin{equation}
\label{eq:assumption}
L^0(\mathbf{x},K\mathbf{x})+F^0(f(\mathbf{x},K\mathbf{x}))-F^0(\mathbf{x}) < 0, \; \forall \mathbf{x} \in \mathbb{X}\setminus \{B_\delta\}. 
\end{equation}
\end{assumption}
\noindent Based on the design of the cost functions in \eqref{eq:costSmallJ}, $L^0(\mathbf{x},K\mathbf{x})$ is the running cost of applying $K\mathbf{x}$ at $\mathbf{x}$, $F^0(f(\mathbf{x},K\mathbf{x}))$ is the terminal cost of the $f(\mathbf{x},K\mathbf{x})$, which is the state from applying $K\mathbf{x}$ on system \eqref{eq:sysmodel} at $\mathbf{x}$, and $F^0(\mathbf{x})$ is the terminal cost of the current state $\mathbf{x}$. It follows that $L^0(\mathbf{x},K\mathbf{x})+F^0(f(\mathbf{x},K\mathbf{x}))-F^0(\mathbf{x})$ can be viewed as the additional cost of the primary destination by applying $K\mathbf{x}$ at $\mathbf{x}$ compared with the cost of $\mathbf{x}$ itself. Therefore, condition \eqref{eq:assumption} implies that for any feasible state $\mathbf{x}$ that is outside $B_\delta$, we can achieve a cost reduction of the primary destination by applying linear feedback $K\mathbf{x}$ at $\mathbf{x}$. Assumption \ref{assumption 2} is also included in \cite{lazar2006stabilizing} to achieve a non-increasing cost function for the stability proof of MPC. 

Before presenting the lemma and theorems, we first define the tuning parameter $\hat{\mathbf{u}}$ from the backup plan constrained MPC algorithm, and a new term $P$ based on $\hat{\mathbf{u}}$.
Given a system function $f$, the running cost $L^i$ and the terminal cost $F^i$ for $i=0,\dots,m$, the control input $\hat{\mathbf{u}}$ used to form $\mathbf{U}_s(\mathbf{x}_k)$ in \eqref{eq:u_s} from the backup plan constrained MPC algorithm is defined as:
\begin{align}
\label{eq:uhat}
    \hat{\mathbf{u}} = \argmin_{\mathbf{u}\in\mathbb{U}}\max_{\{i,\mathbf{x}\}\in \mathbb{I}\times\mathbb{X}}L^i(\mathbf{x},\mathbf{u})+F^i(f(\mathbf{x},\mathbf{u}))-F^i(\mathbf{x}),
\end{align} 
where $\mathbb{I} = \{0,1,2,\dots,m\}$. Based on the control input $\hat{\mathbf{u}}$, constant $P$ is defined as:
\begin{align}
\label{eq:P}
P= \max_{\{i,\mathbf{x}\}\in\mathbb{I}\times\mathbb{X}}L^i(\mathbf{x},\hat{\mathbf{u}})+F^i(f(\mathbf{x},\hat{\mathbf{u}}))-F^i(\mathbf{x})=\min_{\mathbf{u} \in \mathbb{U}}\max_{\{i,\mathbf{x}\}\in \mathbb{I}\times\mathbb{X}}L^i(\mathbf{x},\mathbf{u})+F^i(f(\mathbf{x},\mathbf{u}))-F^i(\mathbf{x}).
\end{align}
According to \eqref{eq:uhat}, $\hat{\mathbf{u}}$ is the feasible control input that results in the smallest value of $\max_{\{i,\mathbf{x}\}\in \mathbb{I}\times\mathbb{X}}L^i(\mathbf{x},\mathbf{u})+F^i(f(\mathbf{x},\mathbf{u}))-F^i(\mathbf{x})$, i.e., the smallest upper bound of $L^i(\mathbf{x},\mathbf{u})+F^i(f(\mathbf{x},\mathbf{u}))-F^i(\mathbf{x})$ over all feasible states and destinations. Since $L^i(\mathbf{x},\hat{\mathbf{u}})+F^i(f(\mathbf{x},\hat{\mathbf{u}}))-F^i(\mathbf{x})$ can be viewed as the additional cost of the alternative destination $i$ by applying $\hat{\mathbf{u}}$ at $\mathbf{x}$ compared with $\mathbf{x}$ itself, we can view $P$ in \eqref{eq:P} as the upper bound of the additional cost of any alternative destination by applying $\hat{\mathbf{u}}$ at any feasible state.
A non-positive $P$ defined in \eqref{eq:P} leads to $L^i(\mathbf{x},\hat{\mathbf{u}})+F^i(f(\mathbf{x},\hat{\mathbf{u}}))-F^i(\mathbf{x}) \leq 0 \; \forall i \; \& \forall \mathbf{x}$, which implies $F^i(f(\mathbf{x},\hat{\mathbf{u}}))-F^i(\mathbf{x}) \leq 0 \; \forall i \; \& \forall \mathbf{x}$ since the running cost $L^i$ is always non-negative. $F^i(f(\mathbf{x},\hat{\mathbf{u}}))-F^i(\mathbf{x}) \leq 0 \; \forall i \; \& \forall \mathbf{x}$ further indicates that by applying $\hat{\mathbf{u}}$ at any feasible state $\mathbf{x}$, the resulting state $f(x,\hat{\mathbf{u}})$ is closer to all destinations compared to $\mathbf{x}$, which is extremely unlikely given the backup plan constrained control problems where we have different destinations. In addition, a non-positive $P$ implies that by applying $\hat{\mathbf{u}}$  at any feasible state, we will not increase the costs of all destinations, which further implies that all the cost functions are minimized at the same state, i.e., we have aligned objectives. Hence, having a non-positive $P$ is a conservative assumption for backup plan constrained control problems.
\cite{bemporad2009multiobjective} proves the asymptotic stability of MMPC with the assumption of $L^i(\mathbf{x},K\mathbf{x})+F^i(f(\mathbf{x},K\mathbf{x}))-F^i(\mathbf{x})\le 0\; \forall \{i,\mathbf{x}\}\in\mathbb{I}\times\mathbb{X}$ (equation 8(a) in \cite{bemporad2009multiobjective}) for some matrix $K$, and by replacing $K\mathbf{x}$ with $\hat{\mathbf{u}}$ for the alternative destinationss, this assumption is identical to having a non-positive $P$ and aligned objectives.
In this paper, in addition to considering a non-positive $P$, we will show that when $P$ is positive, i.e., we have misaligned objectives, the system \eqref{eq:sysmodel} in the closed loop, following the backup plan constrained MPC algorithm, is asymptotically stable.

In the following lemma, we will introduce two desired properties for the tuning parameters of the backup plan constrained MPC algorithm and prove the existence of the tuning parameters satisfying those properties.
\begin{lemma}
\label{lemma:existence}
Let $L^i$ and $F^i, i = 0,1,\cdots,m$ satisfy Assumption~\ref{assumption:1}. Let Assumption \ref{assumption 2} hold. There always exists a set of positive parameters $\delta$, $\gamma^i$ ($i=1,2,\dots,m$), and $\mu$ such that the following two introduced properties hold:
\begin{align}
\label{eq:beta}
    &\beta=\min_{\mathbf{x}\in\mathbb{X}\setminus\{B_\delta\}}\alpha_b^0(\mathbf{x})=\min_{\mathbf{x}\in \mathbb{X}\setminus \{B_\delta\}} 1-\sum_{i=1}^m \gamma^i\frac{\|\mathbf{x}\|}{\max(\mu,\|\mathbf{x}-\mathbf{p}^i\|)}>0,\\
    &\label{eq:al 1}
\beta(L^0(\mathbf{x},K\mathbf{x})+F^0(f(\mathbf{x},K\mathbf{x}))-F^0(\mathbf{x})) \le-(1-\beta) P, \;
\forall \mathbf{x} \in \mathbb{X} \setminus \{B_\delta\},
\end{align}
where $\alpha_b^0(\mathbf{x})$ is defined in \eqref{eq:alpha_b} and matrix $K$ satisfies \eqref{eq:assumption} from Assumption \ref{assumption 2}.
\end{lemma}

\begin{proof}
We first consider $P>0$.
Pick any positive parameter $\delta>0$ s.t. $B_\delta \subseteq \mathbb{X}$. Since $\gamma_i$ and $\mu$ are positive, $\gamma^i\frac{\|\mathbf{x}\|}{\max(\mu,\|\mathbf{x}-\mathbf{p}^i\|)} \le \gamma^i \frac{\|\mathbf{x}\|}{\mu} \le \gamma^i \frac{\max_{\mathbf{x}\in\mathbb{X}\setminus \{B_\delta\}} \|\mathbf{x}\|}{\mu}$ holds $\; \forall (i,\mathbf{x}) \in\mathbb{I}\times\mathbb{X}\setminus \{B_\delta\}$, and it follows that $\sum_{i=1}^m  \gamma^i\frac{\|\mathbf{x}\|}{\max(\mu,\|\mathbf{x}-\mathbf{p}^i\|)}\le \sum_{i=1}^m\gamma^i \frac{\max_{\mathbf{x}\in\mathbb{X}\setminus \{B_\delta\}}\|\mathbf{x}\|}{\mu}$ holds $\forall \mathbf{x} \in \mathbb{X}\setminus \{B_\delta\}$. Let $z = \max_{\mathbf{x}\in\mathbb{X}\setminus \{B_\delta\}}\|\mathbf{x}\| > 0$. It follows that:
\begin{align}
\label{eq:existence of beta 1}
    &1-\sum_{i=1}^m  \gamma^i\frac{\|\mathbf{x}\|}{\max(\mu,\|\mathbf{x}-\mathbf{p}^i\|)} \ge 1-z \frac{\sum_{i=1}^m \gamma^i}{\mu}, \;
    \forall \mathbf{x} \in \mathbb{X}\setminus \{B_\delta\}.
\end{align}
Define $k_1 = \max_{\mathbf{x}\in\mathbb{X}\setminus \{B_\delta\}} L^0(\mathbf{x},K\mathbf{x})+F^0(f(\mathbf{x},K\mathbf{x}))-F^0(\mathbf{x})$. From Assumption \ref{assumption:1}, $L^i, F^i$, and $f$ are continuous. Thus, the function $L^0(\mathbf{x},K\mathbf{x})+F^0(f(\mathbf{x},K\mathbf{x}))-F^0(\mathbf{x})$ is also continuous with respect to $\mathbf{x}$. Since $\mathbb{X}$ is compact, $\mathbb{X}\setminus\{B_\delta\}$ is also compact. It follows that the maximum can be achieved, and that $k_1$ is well-defined based on Extreme Value Theorem \citep{protter2006basic}. Similarly, $\beta, P, \text{and } z$ are also well-defined because all of the domains are compact and functions are continuous. 
Since $K$ satisfies \eqref{eq:assumption}, we have $k_1 < 0$. Considering $k_1 < 0$ and $P>0$, we have $0 <\frac{P}{P-k_1} < 1$. Then, we can always find sufficiently small $\gamma^i$ and sufficiently large $\mu$ s.t. \eqref{eq:beta} holds and $0 < \frac{\sum_{i=1}^m\gamma^i}{\mu} \le \frac{(1-\frac{P}{P-k_1})}{z}$,
which implies
\begin{align}
\label{eq:existence of beta 2}
    1-z\frac{\sum_{i=1}^m \gamma^i}{\mu} \ge \frac{P}{P-k_1}.
\end{align}  
By taking into account of \eqref{eq:existence of beta 1} and \eqref{eq:existence of beta 2}, we have
\begin{align}
\label{eq:existence of beta 3}
    1-\sum_{i=1}^m  \gamma^i\frac{\|\mathbf{x}\|}{\max(\mu,\|\mathbf{x}-\mathbf{p}^i\|)} \ge \frac{P}{P-k_1}, \ \; \forall \mathbf{x} \in \mathbb{X}\setminus \{B_\delta\}.
\end{align}
Then, considering \eqref{eq:existence of beta 3} and definition of $\beta$ in \eqref{eq:beta}, we have:
\begin{align}
    \beta = \min_{\mathbf{x}\in\mathbb{X}\setminus\{B_\delta\}} 1-\sum_{i=1}^m  \gamma^i\frac{\|\mathbf{x}\|}{\max(\mu,\|\mathbf{x}-\mathbf{p}^i\|)} \geq \frac{P}{P-k_1},
\end{align}
which implies $\beta k_1 \le P(\beta-1)$. According to the definition of $k_1$ and the fact that $\beta > 0$, we have $\beta(L^0(\mathbf{x},K\mathbf{x})+F^0(f(\mathbf{x},K\mathbf{x}))-F^0(\mathbf{x})) \le \beta k_1 \le P(\beta-1) \; \forall \mathbf{x} \in \mathbb{X}\setminus \{B_\delta\}$, and we obtain \eqref{eq:al 1}. 

Next, we consider $P \leq 0$. Pick any positive parameter $\delta>0$ s.t. $B_\delta \subseteq \mathbb{X}$. Since $\mathbb{X}$ is a compact set, we can always find sufficiently small $\gamma^i (i=1,2,\cdots,m)$ and sufficiently large $\mu$ s.t. property \eqref{eq:beta} holds. Since $\delta$ and $\gamma^i$ are positive parameters, we have $\beta<1$ according to \eqref{eq:beta}. With $L^0(\mathbf{x},K\mathbf{x})+F^0(f(\mathbf{x},K\mathbf{x}))-F^0(\mathbf{x}) < 0 \; \forall \mathbf{x} \in \mathbb{X}\setminus \{B_\delta\}$ from Assumption \ref{assumption 2}, $0<\beta<1$, and $P\le0$, it follows that \eqref{eq:al 1} holds.

Thus, the existence of $\delta,\gamma^i,\mu$ satisfying \eqref{eq:beta} and \eqref{eq:al 1} is guaranteed.
\end{proof}

\begin{remark}
\label{remark 1}
In \eqref{eq:al 1}, we design $\beta(L^0(\mathbf{x},K\mathbf{x})+F^0(f(\mathbf{x},K\mathbf{x}))-F^0(\mathbf{x})) \le-(1-\beta) P$ to hold for $\mathbb{X} \setminus \{B_\delta\}$ instead of $\mathbb{X}$, and the reason for excluding $B_\delta$ is explained as follows. According to the definition of $\beta$ in \eqref{eq:beta}, when $\beta=1$, we have $\boldsymbol{\alpha}_b(\mathbf{x})=[1,0,0,\dots,0]^\top\;\forall \mathbf{x}\in\mathbb{X}\setminus\{B_\delta\}$, which only considers the primary destination without the backup plan safety. Thus, $\beta=1$ is unfavorable and we need $0<\beta<1$ to include alternative destinations. When we have a positive $P$, the only possible solution to satisfy $\beta(L^0(\mathbf{x},K\mathbf{x})+F^0(f(\mathbf{x},K\mathbf{x}))-F^0(\mathbf{x})) \le-(1-\beta) P$ at $\mathbf{x}=0$ is $\beta = 1$.
In order to ensure the existence of $\beta < 1$, we design $\beta(L^0(\mathbf{x},K\mathbf{x})+F^0(f(\mathbf{x},K\mathbf{x}))-F^0(\mathbf{x})) \le-(1-\beta) P$ to hold for $\mathbf{x}\in\mathbb{X}\setminus\{B_\delta\}$, excluding the origin. 
 
\end{remark}

As we previously explained, having a non-positive $P$ implies aligned objectives, i.e., the cost functions of all destinations are minimized at the same state, and thus it is a conservative assumption. In the following theorem, we will specifically consider $P>0$, and show that when the tuning parameters of the backup plan constrained MPC algorithm satisfy the two properties introduced in Lemma \ref{lemma:existence}, we can achieve a desired property related to the overall cost, which will be leveraged to show the closed-loop asymptotic stability. 
Before presenting the theorem, we need the following definitions.
Given any design of the weight vector $\boldsymbol{\alpha}(\mathbf{x}_k)$, define $\mathbf{U}(\mathbf{x}_k,\boldsymbol{\alpha}(\mathbf{x}_k)) = [\mathbf{U}^0(\mathbf{x}_k,\boldsymbol{\alpha}(\mathbf{x}_k))^\top, \mathbf{U}^1(\mathbf{x}_k,\boldsymbol{\alpha}(\mathbf{x}_k)),\dots,\mathbf{U}^m(\mathbf{x}_k,\boldsymbol{\alpha}(\mathbf{x}_k))]^\top$ as the solution to the feasibility maximization \eqref{eq:mmpc} for $\mathbf{x}=\mathbf{x}_k$ and $\boldsymbol{\alpha}=\boldsymbol{\alpha}(\mathbf{x}_k)$. Given any design of the weight vector $\alpha(\mathbf{x}_k)$ and a positive constant $\boldsymbol{\alpha}$, denote the set $\mathbb{\hat{X}}(\boldsymbol{\alpha}(\mathbf{x}_k),\boldsymbol{\alpha})$ as:
\begin{align}
\label{eq:xhat}
\mathbb{\hat{X}}(\alpha(\mathbf{x}_k),\delta) = \{\mathbf{x}_k \mid \mathbf{x}_k \in\mathbb{X}\setminus \{B_\delta\} \;\&\; \mathbf{x}_{k,f}\in \mathbb{X}\setminus \{B_\delta\}\},
\end{align}
where $B_\delta$ is an open ball centered at the origin with radius $\delta$ and $\mathbf{x}_{k,f}$ is the resulting state from simulating $\mathbf{U}^0(\mathbf{x}_k,\alpha(\mathbf{x}_k))$ from $\mathbf{U}(\mathbf{x}_k,\alpha(\mathbf{x}_k))$ on system \eqref{eq:sysmodel} at $\mathbf{x}_k$. In other words, for $\mathbf{x}_k\in\mathbb{\hat{X}}(\alpha(\mathbf{x}_k),\delta)$, both $\mathbf{x}_k$ and its final state toward the primary objective $\mathbf{x}_{k,f}$ are outside the ball $B_\delta$ around the origin, where $\mathbf{x}_{k,f}$ is related to the solution to the feasibility maximization problem \eqref{eq:mmpc} for $\alpha=\alpha(\mathbf{x}_k)$.

\begin{theorem}
\label{theorem1}
Let $L^i$ and $F^i, i = 0,1,\cdots,m$, satisfy Assumption~\ref{assumption:1}. Let Assumption \ref{assumption 2} hold. When $P> 0$, for the design of the weight vector in the feasibility maximization problem \eqref{eq:mmpc}, if we set $\boldsymbol{\alpha}_*(x)=\boldsymbol{\alpha}_b(\mathbf{x}) \; \forall \mathbf{x} \in \mathbb{X}$ where $\boldsymbol{\alpha}_b(\mathbf{x})$ is defined in \eqref{eq:alpha_b} with positive parameters $\delta$, $\gamma^i$ ($i=1,2,\dots,m$), and $\mu$ satisfying conditions \eqref{eq:beta} and \eqref{eq:al 1}, then the following condition holds:
    \begin{equation}
    \label{eq:2nd inequality}
    \begin{split}
        \boldsymbol{\alpha}_*(\mathbf{x}_k)^\top \mathbf{J}(\mathbf{x}_{k+1},\mathbf{U}_s(\mathbf{x}_{k+1})) \le \boldsymbol{\alpha}_*(\mathbf{x}_k)^\top \mathbf{J}(\mathbf{x}_k,\mathbf{U}_*(\mathbf{x}_k)), \;
        \forall \mathbf{x}_k \in \mathbb{\hat{X}}(\boldsymbol{\alpha}_*(\mathbf{x}_k), \delta)
        \end{split},
    \end{equation}
where $\mathbf{U}_*(\mathbf{x}_k)$ is the solution to the feasibility maximization problem \eqref{eq:mmpc} for $\mathbf{x}=\mathbf{x}_k$ and $\boldsymbol{\alpha}=\boldsymbol{\alpha}_*(\mathbf{x}_k)$, $\mathbf{x}_{k+1}$ is the state by simulating the first element of $\mathbf{U}_*(\mathbf{x}_k)$ on system \eqref{eq:sysmodel} at $\mathbf{x}_k$, $\mathbb{\hat{X}}(\boldsymbol{\alpha}_*(\mathbf{x}_k),\delta)$ is defined in \eqref{eq:xhat}, and $\mathbf{U}_s(\mathbf{x}_{k+1})$ is defined in \eqref{eq:u_s}, constructed with matrix $K$ satisfying \eqref{eq:assumption} and $\hat{\mathbf{u}}$ defined in \eqref{eq:uhat}.
\end{theorem}
Before stating the main proof, we summarize the flow of it in a few steps. Since $\beta$ is defined as the minimum of $\alpha_b^0(\mathbf{x}_k)$, $(1-\beta)$ can be viewed as the maximum of the weight for any alternative destination $\alpha_b^i(\mathbf{x}_k)$. As $P$ indicates an upper bound of the additional cost of any alternative destination by applying $\hat{\mathbf{u}}$ at any feasible state, condition \eqref{eq:al 1} implies that by applying $\mathbf{U}_s(\mathbf{x}_k)$ defined in \eqref{eq:u_s} with matrix $K$ and control input $\hat{\mathbf{u}}$, the weighted decrease in the cost function of the primary destination dominates the weighted increase in the cost functions of the alternative destinations, which leads to \eqref{eq:2nd inequality}.
\begin{proof}
Since Lemma \ref{lemma:existence} and Theorem \ref{theorem1} share the same assumptions, Lemma \ref{lemma:existence} holds, and hence we are always able to find a set of positive parameters $\delta$, $\gamma^i$, and $\mu$ satisfying conditions \eqref{eq:beta} and \eqref{eq:al 1}.
Since $\delta$, $\gamma^i$, and $\mu$ are all positive, we have $0< \beta < 1$ from the definition of $\beta$ in \eqref{eq:beta}, and $\alpha_*^i(\mathbf{x})$ satisfies $\alpha_*^i \in [0,1] \; \forall i$ and $\sum_{i=0}^{m}\alpha_*^i=1$ $\forall \{i,x\} \in \mathbb{I} \times \mathbb{X}\setminus \{B_\delta\}$ based on the definition of $\boldsymbol{\alpha}_b(\mathbf{x}_k)$ in \eqref{eq:alpha_b}. 
Additionally, we have $\alpha_*^0(\mathbf{x})=\alpha_b^0(\mathbf{x})\ge \beta > 0 \; \forall \mathbf{x} \in \mathbb{X}\setminus \{B_\delta\}$.
Since $L^0(\mathbf{x},K\mathbf{x})+F^0(f(\mathbf{x},K\mathbf{x}))-F^0(\mathbf{x}) < 0 \; \forall \mathbf{x} \in \mathbb{X}\setminus\{B_\delta\}$ from \eqref{eq:assumption} in Assumption \ref{assumption 2} and $\alpha_*^0(\mathbf{x}) \ge \beta \; \forall\; \mathbf{x} \in \mathbb{X}\setminus \{B_\delta\}$ from \eqref{eq:beta}, we have:
\begin{align}
\label{eq:LHS_1}
    &\alpha_*^0(\mathbf{x}_k)(L^0(\mathbf{x}_{k,f},K\mathbf{x}_{k,f})
    +F^0(f(\mathbf{x}_{k,f},K\mathbf{x}_{k,f}))-F^0(\mathbf{x}_{k,f}))\nnum\\
    &\le \beta(L^0(\mathbf{x}_{k,f},K\mathbf{x}_{k,f})
     +F^0(f(\mathbf{x}_{k,f},K\mathbf{x}_{k,f}))-F^0(\mathbf{x}_{k,f})),\nnum\\
     & \forall \mathbf{x}_k \in \mathbb{X}\setminus\{B_\delta\}, \mathbf{x}_{k,f}\in \mathbb{X}\setminus\{B_\delta\}.
\end{align}
According to the definition of $P$ in \eqref{eq:P}, we have
\begin{align}
  \label{eq:RHS_3}
  &P \ge L^i(\mathbf{x},\hat{\mathbf{u}})+F^i(f(\mathbf{x},\hat{\mathbf{u}}))-F^i(\mathbf{x}),\;
  \forall \{i,\mathbf{x}\}\in\mathbb{I} \times \mathbb{X}.
\end{align}
Following \eqref{eq:RHS_3}, we have
\begin{align}
    \label{eq:RHS_4}
    &P \ge \frac{1}{N-1} \sum_{p=0}^{N-2} L^i(\mathbf{x}_{p},\hat{\mathbf{u}})+F^i(f(\mathbf{x}_{p},\hat{\mathbf{u}}))-F^i(\mathbf{x}_{p}),
    \;\forall \{i,\mathbf{x}_{p}\}\in\mathbb{I} \times \mathbb{X},
\end{align}
which is the averaged value version of \eqref{eq:RHS_3}.
Now, since $\alpha_*^i(\mathbf{x}) \ge 0 \; \forall \{i,\mathbf{x}\} \in \mathbb{I} \times \mathbb{X}\setminus \{B_\delta\}$, it follows from \eqref{eq:RHS_4} that
\begin{align}
\label{eq:RHS_5}
    -\sum_{i=1}^m\alpha_*^i(\mathbf{x}_k) P \le -\sum_{i=1}^m\alpha_*^i(\mathbf{x}_k) \frac{1}{N-1}
    \sum_{p=0}^{N-2} L^i(\mathbf{x}_{p},\hat{\mathbf{u}})+F^i(f(\mathbf{x}_{p},\hat{\mathbf{u}}))-F^i(\mathbf{x}_{p}), \;
    \forall \{\mathbf{x}_k,\mathbf{x}_{p}\} \in \mathbb{X}\setminus\{B_\delta\}\times \mathbb{X}.
\end{align}
By the definition of $\boldsymbol{\alpha}_b(\mathbf{x}_k)$ in \eqref{eq:alpha_b} and $\beta$ in \eqref{eq:beta}, with $\boldsymbol{\alpha}_*(\mathbf{x}_k)=\boldsymbol{\alpha}_b(\mathbf{x}_k)$, it follows that 
\begin{align}
\label{eq:alpha_i bound}
\sum_{i=1}^m \alpha_*^i(\mathbf{x}_k) \le (1-\beta), \; \forall \mathbf{x}_k\in \mathbb{X}\setminus\{B_\delta\}.
\end{align}
Considering \eqref{eq:alpha_i bound} and the fact that $P>0$, we have
\begin{equation}
\label{eq:RHS_6}
-(1-\beta) P \le -\sum_{i=1}^m \alpha_*^i(\mathbf{x}_k) P, \ \; \forall \mathbf{x}_k \in \mathbb{X}\setminus\{B_\delta\}.
\end{equation}
Combining \eqref{eq:RHS_5} and \eqref{eq:RHS_6}, we achieve
\begin{align}
\label{eq:RHS}
-(1-\beta) P \le- \sum_{i=1}^m\alpha_*^i(\mathbf{x}_k)\frac{1}{N-1}\sum_{p=0}^{N-2}L^i(\mathbf{x}_{p},\hat{\mathbf{u}})\!+\!F^i(f(\mathbf{x}_{p},\hat{\mathbf{u}}))\!-\!F^i(\mathbf{x}_{p}), \;
\forall \{\mathbf{x}_k,\mathbf{x}_{p}\} \in \mathbb{X}\setminus\{B_\delta\}\times \mathbb{X}.
\end{align}
Following the inequalities \eqref{eq:LHS_1}, \eqref{eq:al 1}, and \eqref{eq:RHS} in order, we can relate the left-hand side of \eqref{eq:LHS_1} and right-hand side of \eqref{eq:RHS} and get
\begin{align}
\label{eq:multi}
&\alpha_*^0(\mathbf{x}_k)(L^0(\mathbf{x}_{k,f},K\mathbf{x}_{k,f})+F^0(f(\mathbf{x}_{k,f},K\mathbf{x}_{k,f}))-F^0(\mathbf{x}_{k,f}))\nnum\\ &\le -\sum_{i=1}^m\alpha_*^i(\mathbf{x}_k)\frac{1}{N-1}\sum_{p=0}^{N-2}L^i(\mathbf{x}_p,\hat{\mathbf{u}})\!+\!F^i(f(\mathbf{x}_p,\hat{\mathbf{u}}))\!-\!F^i(\mathbf{x}_p), \nnum \\
&\forall \{\mathbf{x}_k,\mathbf{x}_{k,f},\mathbf{x}_p\} \in \mathbb{X}\setminus\{B_\delta\}\times \mathbb{X}\setminus\{B_\delta\}\times \mathbb{X}.
\end{align}
Inequality \eqref{eq:2nd inequality} in backup plan constrained control problems can be reformulated as:
\begin{align}
\label{eq:proof end}
    &\boldsymbol{\alpha}_*(\mathbf{x}_k)^\top \mathbf{J}(\mathbf{x}_{k+1},\mathbf{U}_s(\mathbf{x}_{k+1})) - \boldsymbol{\alpha}_*(\mathbf{x}_k)^\top \mathbf{J}(\mathbf{x}_k,\mathbf{U}_*(\mathbf{x}_k)) \nnum\\
    &=\sum_{i=0}^m \alpha_*^i(\mathbf{x}_k)[\mathbf{J}^i(\mathbf{x}_{k+1},\mathbf{U}_s^i(\mathbf{x}_{k+1}))-\mathbf{J}^i(\mathbf{x}_k,{\mathbf{U}}_*^i(\mathbf{x}_k))]
    \nnum\\
    &=\alpha_*^0(\mathbf{x}_k)[\mathbf{J}^0(\mathbf{x}_{k+1},{\mathbf{U}}_s^0(\mathbf{x}_{k+1}))-\mathbf{J}^0(\mathbf{x}_k,{\mathbf{U}}_*^0(\mathbf{x}_k))]\nnum\\
    &+ \sum_{i=1}^m \alpha_*^i(\mathbf{x}_k)[\mathbf{J}^i(\mathbf{x}_{k+1},\mathbf{U}_s^i(\mathbf{x}_{k+1}))-\mathbf{J}^i(\mathbf{x}_k,{\mathbf{U}}_*^i(\mathbf{x}_k))]
    \nnum\\
    & =\alpha_*^0(\mathbf{x}_k)[J^0(\mathbf{x}_{k+1},\mathbf{U}_s^0(\mathbf{x}_{k+1}))-J^0(\mathbf{x}_k,\mathbf{U}_*^0(\mathbf{x}_k))]\nnum\\  
    &+ \sum_{i=1}^m \alpha_*^i(\mathbf{x}_k)\frac{1}{N-1}\sum_{p=0}^{N-2}J^i(\mathbf{x}_{k+1},\mathbf{U}^i_{s,p}(\mathbf{x}_{k+1}))-J^i(\mathbf{x}_k,\mathbf{U}^i_{*,p}(\mathbf{x}_k))
    \nnum\\
    &=\alpha_*^0(\mathbf{x}_k)(L^0(\mathbf{x}_{k,f},K\mathbf{x}_{k,f})+ F^0(f(\mathbf{x}_{k,f},K\mathbf{x}_{k,f}))-F^0(\mathbf{x}_{k,f})-
    L^0(\mathbf{x}_k,\mathbf{u}^*_k))  \nnum \\
    &+\sum_{i=1}^m\alpha_*^i(\mathbf{x}_k)\frac{1}{N-1} \sum_{p=0}^{N-2}L^i(\mathbf{x}_{i,p},\hat{\mathbf{u}})\!+\!F^i(f(\mathbf{x}_{i,p},\hat{\mathbf{u}}))\!-\!F^i(\mathbf{x}_{i,p})-L^i(\mathbf{x}_k,\mathbf{u}^*_k), 
\end{align}
where $\mathbf{u}^*_k$ is the first element of $\mathbf{U}_*^0(\mathbf{x}_k)$, $\mathbf{x}_{k,f}$ is the final state of applying $\mathbf{U}_*^0(\mathbf{x}_k)$ at $\mathbf{x}_k$, and the point $\mathbf{x}_{i,p}$ is the final state from the application of
$\mathbf{U}^i_{*,p}(\mathbf{x}_k)$ at $\mathbf{x}_k$. For $\mathbf{x}_k\in \hat{\mathbb{X}}(\boldsymbol{\alpha}_*(\mathbf{x}_k),\delta)$, according to the definition of $\hat{\mathbb{X}}(\boldsymbol{\alpha}_*(\mathbf{x}_k),\delta)$ in \eqref{eq:xhat}, we have $\mathbf{x}_k \in \mathbb{X}\setminus \{B_\delta\}$ and $\mathbf{x}_{k,f}\in \mathbb{X}\setminus \{B_\delta\}$. Since $\mathbf{U}_*(\mathbf{x}_k)$ is a feasible control input of the feasibility maximization problem \eqref{eq:mmpc}, we also have $\mathbf{x}_{i,p}\in \mathbb{X}$, where $i=1,2,\dots,m$ and $p=0,1,\dots,N-2$.
By following \eqref{eq:proof end} and applying the results in \eqref{eq:multi}, we achieve
\begin{align}
    \label{eq:difference}
    &\boldsymbol{\alpha}_*(\mathbf{x}_k)^\top \mathbf{J}(\mathbf{x}_{k+1},\mathbf{U}_s(\mathbf{x}_{k+1})) - \boldsymbol{\alpha}_*(\mathbf{x}_k)^\top \mathbf{J}(\mathbf{x}_k,\mathbf{U}_*(\mathbf{x}_k)) \le -\sum_{i=0}^m \alpha_*^i(\mathbf{x}_k)L^i(\mathbf{x}_k,\mathbf{u}^*_k), \; \forall \mathbf{x}_k\in \hat{\mathbb{X}}(\boldsymbol{\alpha}_*(\mathbf{x}_k),\delta).
\end{align}
Since $L^i(\cdot) \ge 0$ and $\alpha_*^i(\cdot)\ge0 \; \forall i$, it follows from \eqref{eq:difference} that \eqref{eq:2nd inequality} holds.
\end{proof}

Theorem \ref{theorem1} states that when $P>0$ and $\boldsymbol{\alpha}_*(\mathbf{x}_k)$ is formulated as the baseline weight vector \eqref{eq:alpha_b} with parameters satisfying \eqref{eq:beta} and \eqref{eq:al 1}, property \eqref{eq:2nd inequality} holds.
Next, we will introduce a theorem similar to Theorem \ref{theorem1} when we have $P\le0$.
\begin{theorem}
\label{theorem2}
Let $L^i$ and $F^i, i = 0,1,\cdots,m$, satisfy Assumption~\ref{assumption:1}. Let Assumption \ref{assumption 2} hold. When $P\le0$, for any design of $\boldsymbol{\alpha}_*(\mathbf{x}_k)$ in the feasibility maximization problem \eqref{eq:mmpc} satisfying $\alpha^i_* \in [0,1] \; \forall i$ and $\sum_{i=0}^{m}\alpha^i_*=1$, the following property holds:
\begin{equation}
    \label{eq:2nd inequality 2}
    \begin{split}
        \boldsymbol{\alpha}_*(\mathbf{x}_k)^\top \mathbf{J}(\mathbf{x}_{k+1},\mathbf{U}_s(\mathbf{x}_{k+1})) \le \boldsymbol{\alpha}_*(\mathbf{x}_k)^\top \mathbf{J}(\mathbf{x}_k,\mathbf{U}_*(\mathbf{x}_k)), \;
        \forall \mathbf{x}_k \in \mathbb{X}\setminus \{B_\delta\}
        \end{split},
    \end{equation}
where $\mathbf{U}_*(\mathbf{x}_k)$ is the solution to the feasibility maximization problem \eqref{eq:mmpc} for $\mathbf{x}=\mathbf{x}_k$ and $\boldsymbol{\alpha}=\boldsymbol{\alpha}_*(\mathbf{x}_k)$, $\mathbf{x}_{k+1}$ is the state by simulating the first element of $\mathbf{U}_*(\mathbf{x}_k)$ on system \eqref{eq:sysmodel} at $\mathbf{x}_k$, and $\mathbf{U}_s(\mathbf{x}_{k+1})$ is defined in \eqref{eq:u_s}, constructed with matrix $K$ satisfying \eqref{eq:assumption} and $\hat{\mathbf{u}}$ defined in \eqref{eq:uhat}.
\end{theorem}
\begin{proof}
Since $P\le 0$, according to its definition in \eqref{eq:P}, it follows that 
\begin{align}
\label{eq:theorem 3}
   L^i(\mathbf{x},\hat{\mathbf{u}})+F^i(f(\mathbf{x},\hat{\mathbf{u}}))-F^i(\mathbf{x})\le P \le 0\;
  \forall \{i,\mathbf{x}\}\in\mathbb{I} \times \mathbb{X}.
\end{align}
By taking into account \eqref{eq:theorem 3}, \eqref{eq:assumption} from Assumption \ref{assumption 2}, and \eqref{eq:proof end}, it follows that \eqref{eq:difference} still holds. Since $L^i(\cdot) \ge 0$ and $\alpha_*^i(\cdot)\ge0 \; \forall i$, we can achieve \eqref{eq:2nd inequality 2} for any feasible design of $\boldsymbol{\alpha}_*(\mathbf{x}_k)$ from \eqref{eq:difference}.
\end{proof}
Next, we will use the result of Theorems \ref{theorem1} and \ref{theorem2} to prove the asymptotic stability of the closed-loop system for both $P\le 0$ and $P > 0$.
\begin{theorem}
\label{theorem3}
Let $L^i$ and $F^i, i = 0,1,\cdots,m$ satisfy Assumption~\ref{assumption:1}. Let Assumption \ref{assumption 2} hold. If the tuning parameters $\delta$, $\gamma^i$ ($i=1,2,\dots,m$), and $\mu$ satisfy conditions \eqref{eq:beta} and \eqref{eq:al 1}, $K$ satisfies \eqref{eq:assumption} and $\hat{\mathbf{u}}$ is defined in \eqref{eq:uhat}, then the system \eqref{eq:sysmodel} in a closed-loop with the backup plan constrained MPC algorithm defined in Algorithm \ref{al:backup} is asymptotically stable with respect to $\mathbf{p}^0$ for any initial condition $\mathbf{x}_0 \in \mathbb{X}$. 
\end{theorem}
Before stating the proof, we present the following definitions.
Based on the definition in \eqref{eq:xhat}, we have $\mathbb{\hat{X}}(\boldsymbol{\alpha}_t(\mathbf{x}_k),\delta) = \{\mathbf{x}_k \mid \mathbf{x}_k \in\mathbb{X}\setminus \{B_\delta\} \;\&\; \mathbf{x}_{k,f}\in \mathbb{X}\setminus \{B_\delta\}\}$, where $\boldsymbol{\alpha}_t(\mathbf{x}_k)$ is defined in \eqref{eq:alpha_t}, $\mathbf{x}_{k,f}$ is the resulting state from simulating $\mathbf{U}^0(\mathbf{x}_k,\boldsymbol{\alpha}_t(\mathbf{x}_k))$ on system \eqref{eq:sysmodel} at $\mathbf{x}_k$ with $\mathbf{U}^0(\mathbf{x}_k,\boldsymbol{\alpha}_t(\mathbf{x}_k))$ from $\mathbf{U}(\mathbf{x}_k,\boldsymbol{\alpha}_t(\mathbf{x}_k))$, which is the solution to the feasibility maximization problem \eqref{eq:mmpc} for $\mathbf{x}=\mathbf{x}_k$ and $\boldsymbol{\alpha}=\boldsymbol{\alpha}_t(\mathbf{x}_k)$. Given any initial condition $\mathbf{x}_0$, let $\mathbf{x}_k$ represent the state of the system following the backup plan constrained MPC algorithm at time $k$. We denote the period while $\mathbf{x}_k\in \mathbb{\hat{X}}(\boldsymbol{\alpha}_t(\mathbf{x}_k),\delta)$ as phase 1 and denote the period after the system state leaves $\mathbb{\hat{X}}(\boldsymbol{\alpha}_t(\mathbf{x}_k),\delta)$ for the first time as phase 2. Thus, according to the backup plan constrained MPC algorithm defined in Algorithm \ref{al:backup}, during phase 1 with $\mathbf{x}_k\in\mathbb{\hat{X}}(\boldsymbol{\alpha}_t(\mathbf{x}_k),\delta)$, we have $\boldsymbol{\alpha}_*(\mathbf{x}_k)=\boldsymbol{\alpha}_t(\mathbf{x}_k)$; during phase 2, we have $\boldsymbol{\alpha}_*(x)=[1,0,\dots,0]^\top \;\forall \mathbf{x}\in \mathbb{X}$. Given any state $\mathbf{x}_k$ and $\boldsymbol{\alpha}_*(\mathbf{x}_k)$ determined by the backup plan constrained MPC algorithm, the value function $V(\mathbf{x}_k)$ is defined as:
\begin{align}
\label{eq:V}
    V(\mathbf{x}_k)={\boldsymbol{\alpha}_*(\mathbf{x}_k)}^\top\mathbf{J}(\mathbf{x}_k,\mathbf{U}_*(\mathbf{x}_k)),
\end{align}
which indicates the overall cost of applying the optimal control input sequence $\mathbf{U}_*(\mathbf{x}_k)$ at state $\mathbf{x}_k$ with the weight vector $\boldsymbol{\alpha}_*(\mathbf{x}_k)$. As we previously explained in Section \ref{sec:intro}, a non-increasing overall cost can lead to the stability of the closed-loop system. For the backup plan constrained control problems, we will show that with $\boldsymbol{\alpha}_*(\mathbf{x}_k)$ determined by the backup plan constrained MPC algorithm, the value function is also non-increasing.

\begin{proof}
With Assumptions \ref{assumption:1} and \ref{assumption 2} satisfied, the existence of parameters $\delta$, $\gamma^i$ ($i=1,2,\dots,m$), and $\mu$ satisfying \eqref{eq:beta} and \eqref{eq:al 1} is guaranteed by Lemma \ref{lemma:existence}.
Here, we prove the asymptotic stability of the closed-loop system with respect to the origin for any initial condition $\mathbf{x}_0\in \mathbb{X}$ by showing: (I) the Lyapunov stability of the closed-loop system and (II) the convergence of the closed-loop system state to the origin for any initial condition $\mathbf{x}_0\in\mathbb{X}$. The convergence in (II) is further proven by (II.i) for any initial condition $\mathbf{x}_0\in \mathbb{X}$ the system enters phase 2 in finite time, and (II.ii) once the system is in phase 2 the closed-loop system state converges to origin.
We will first show (II.i), then (II.ii), and finally (I).

(II.i) When $\mathbf{x}_0\in\mathbb{X}\setminus \mathbb{\hat{X}}(\boldsymbol{\alpha}_t(\mathbf{x}_k),\delta)$, the system starts in phase 2 directly, i.e., the system enters phase 2 in finite time. To prove that the system enters phase 2 in finite time when $\mathbf{x}_0\in\mathbb{\hat{X}}(\boldsymbol{\alpha}_t(\mathbf{x}_k),\delta)$, we will show a non-increasing value function $V$ along the system trajectory following the backup plan constrained algorithm for both $P>0$ and $P\le 0$.

First, consider $P > 0$.
According to the design of $\boldsymbol{\alpha}_t(x)$ in \eqref{eq:alpha_t}, and since $\boldsymbol{\alpha}_*(\mathbf{x}_k)=\boldsymbol{\alpha}_t(\mathbf{x}_k)$ while $\mathbf{x}_k \in \mathbb{\hat{X}}(\boldsymbol{\alpha}_t(\mathbf{x}_k),\delta)$, we have
\begin{align}
\label{eq:first inequality}
    &\boldsymbol{\alpha}_*(\mathbf{x}_{k})^\top \mathbf{J}(\mathbf{x}_{k},\mathbf{U}_s(\mathbf{x}_{k})) 
    \le \boldsymbol{\alpha}_*(\mathbf{x}_{k-1})^\top \mathbf{J}(\mathbf{x}_{k},\mathbf{U}_s(\mathbf{x}_{k})), \text{while } \mathbf{x}_k \in \mathbb{\hat{X}}(\boldsymbol{\alpha}_t(\mathbf{x}_k),\delta) \text{ and system is in phase 1}.
\end{align}
According to the backup plan constrained MPC algorithm and the definition of $\boldsymbol{\alpha}_t(x)$ in \eqref{eq:alpha_t}, for $\mathbf{x}_k\in \mathbb{X}\setminus\{B_\delta\}$, $\alpha_t^0(\mathbf{x}_k)$ is either $\alpha_b^0(\mathbf{x}_k)$ or $\alpha_*^0(\mathbf{x}_{k-1})$, and $\alpha_*^0(\mathbf{x}_{k-1})$ is either $\alpha_t^0(\mathbf{x}_{k-1})$ or $1$. Recursively, $\alpha_t^0(\mathbf{x}_{k-1})$ is either $\alpha_b^0(\mathbf{x}_{k-1})$ or $\alpha_*^0(\mathbf{x}_{k-2})$.
It follows that for $\mathbf{x}_k\in \mathbb{X}\setminus\{B_\delta\}$, the range of $\alpha_t^0(\mathbf{x}_k)$ is the union of set $\{1\}$ and a subset of the range of $\alpha_b^0(\mathbf{x}_k)$. Thus, we have 
\begin{align}
\label{eq:alpha_tandb}
    \min_{\mathbf{x}_k\in\mathbb{X}\setminus\{B_\delta\}} \alpha_t^0(\mathbf{x}_k) \ge \min_{\mathbf{x}_k\in\mathbb{X}\setminus\{B_\delta\}} \alpha_b^0(\mathbf{x}_k).
\end{align}
Similar to \eqref{eq:beta}, let us define
\begin{align}
\label{eq:beta_s}
    \beta_t = \min_{\mathbf{x}_k\in\mathbb{X}\setminus\{B_\delta\}} \alpha_t^0(\mathbf{x}_k).
\end{align}
Considering \eqref{eq:beta_s}, \eqref{eq:alpha_tandb}, and \eqref{eq:beta}, we achieve
\begin{align}
\label{eq:2 beta}
\beta_t= \min_{\mathbf{x}_k\in\mathbb{X}\setminus\{B_\delta\}}\alpha_t^0(\mathbf{x}_k)\ge \min_{\mathbf{x}_k\in\mathbb{X}\setminus\{B_\delta\}}\alpha_b^0(\mathbf{x}_k)= \beta,
\end{align}
where $\beta$ is defined in \eqref{eq:beta} depending on tuning parameters $\delta, \gamma_i, \mu$.
Considering \eqref{eq:2 beta} and \eqref{eq:assumption} from Assumption \ref{assumption 2}, we have:
\begin{align}
\label{eq:theorem3_1}
    &\beta_t(L^0(\mathbf{x},K\mathbf{x})+F^0(f(\mathbf{x},K\mathbf{x}))-F^0(\mathbf{x})) \le \beta(L^0(\mathbf{x},K\mathbf{x})+F^0(f(\mathbf{x},K\mathbf{x}))-F^0(\mathbf{x})), \;\forall \mathbf{x} \in \mathbb{X}\setminus\{B_\delta\}.    
\end{align}
By taking into account \eqref{eq:2 beta} and the assumption that $P>0$, it follows that
\begin{equation}
\label{eq:theorem3_2}
    -(1-\beta)P \le -(1-\beta_t)P.
\end{equation}
Finally, we can relate the left hand side of \eqref{eq:theorem3_1} and the right hand side of \eqref{eq:theorem3_2} using \eqref{eq:al 1} and obtain:
\begin{align}
\label{eq:new_beta}
&\beta_t(L^0(\mathbf{x},K\mathbf{x})+F^0(f(\mathbf{x},K\mathbf{x}))-F^0(\mathbf{x}))\le-(1-\beta_t) P,\;\forall \mathbf{x} \in \mathbb{X}\setminus \{B_\delta\},
\end{align}
which is the $\beta_t$ version of \eqref{eq:al 1}.
Considering \eqref{eq:new_beta} and the fact that for $x \in \mathbb{\hat{X}}(\boldsymbol{\alpha}_t(\mathbf{x}_k), \delta)$ we have $\boldsymbol{\alpha}_*(x)=\boldsymbol{\alpha}_t(x)$, we can follow the same proof of Theorem \ref{theorem1} to prove \eqref{eq:2nd inequality} with the design of $\boldsymbol{\alpha}_*(\mathbf{x}_k)$ from Algorithm \ref{al:backup}, and obtain:
\begin{align}
\label{eq:new 2nd inequality}
    &\boldsymbol{\alpha}_*(\mathbf{x}_{k-1})^\top \mathbf{J}(\mathbf{x}_{k},\mathbf{U}_s(\mathbf{x}_{k})) \le \boldsymbol{\alpha}_*(\mathbf{x}_{k-1})^\top \mathbf{J}(\mathbf{x}_{k-1},\mathbf{U}_*(\mathbf{x}_{k-1})),
    \text{while } \mathbf{x}_{k-1} \in \mathbb{\hat{X}}(\boldsymbol{\alpha}_t(\mathbf{x}_k), \delta).
\end{align}
From optimality, we have
\begin{align}
\label{eq:v(k+1)}
    &V(\mathbf{x}_{k}) = \boldsymbol{\alpha}_*(\mathbf{x}_{k})^\top\mathbf{J}(\mathbf{x}_{k},\mathbf{U}_*(\mathbf{x}_{k}))  \le \boldsymbol{\alpha}_*(\mathbf{x}_{k})^\top\mathbf{J}(\mathbf{x}_{k},\mathbf{U}_s(\mathbf{x}_{k})).
\end{align}
According to Algorithm \ref{al:backup}, when the system at time $k$ is still in phase 1, we have $\mathbf{x}_{k-1} \in \mathbb{\hat{X}}(\boldsymbol{\alpha}_t(\mathbf{x}_k),\delta)$.
Following the inequalities \eqref{eq:v(k+1)}, \eqref{eq:first inequality}, and \eqref{eq:new 2nd inequality} in order, we can relate the left-hand side of \eqref{eq:v(k+1)} and right-hand side of \eqref{eq:new 2nd inequality}, and arrive at
\begin{align}
\label{eq:V decrease}
&V(\mathbf{x}_{k}) \le \boldsymbol{\alpha}_*(\mathbf{x}_{k})^\top \mathbf{J}(\mathbf{x}_{k},\mathbf{U}_s(\mathbf{x}_{k}))\le \boldsymbol{\alpha}_*(\mathbf{x}_{k-1})\mathbf{J}(\mathbf{x}_k,\mathbf{U}_s(\mathbf{x}_k))\nnum\\
&\le \boldsymbol{\alpha}_*(\mathbf{x}_{k-1})^\top \mathbf{J}(\mathbf{x}_{k-1},\mathbf{U}_*(\mathbf{x}_{k-1})) = V(\mathbf{x}_{k-1}),\nnum\\
&\text{while } \mathbf{x}_k \in \mathbb{\hat{X}}(\boldsymbol{\alpha}_t(\mathbf{x}_k),\delta) \text{ and system is in phase 1}.
\end{align}
Now consider $P\le 0$. As Theorem \ref{theorem3} and Theorem \ref{theorem2} share the same assumptions, Theorem \ref{theorem2} holds and we have \eqref{eq:2nd inequality 2}. According to \eqref{eq:xhat}, $\mathbb{\hat{X}}(\boldsymbol{\alpha}_t(\mathbf{x}_k),\delta)$ is a subset of $\mathbb{X}\setminus\{B_\delta\}$. Taking into account the previous result and by replacing $\mathbf{x}_k$ with $\mathbf{x}_{k-1}$ and $\mathbf{x}_{k+1}$ with $\mathbf{x}_k$ for  \eqref{eq:2nd inequality 2}, \eqref{eq:new 2nd inequality} holds. Since \eqref{eq:v(k+1)} and \eqref{eq:first inequality} are properties not related to $P$, they still hold when $P\le 0$. Following the inequalities \eqref{eq:v(k+1)}, \eqref{eq:first inequality}, and \eqref{eq:new 2nd inequality} in order, we are able to achieve \eqref{eq:V decrease}. Till now, we have shown \eqref{eq:V decrease}, a non-increasing value function $V$ while the system is in phase 1, for both $P>0$ and $P\le 0$.

Using \eqref{eq:V decrease}, we next show that with initial condition $\mathbf{x}_0\in\mathbb{\hat{X}}(\boldsymbol{\alpha}_t(\mathbf{x}_k), \delta)$, i.e. when the system starts in phase 1, the state of the closed-loop system will leave $\mathbb{\hat{X}}(\boldsymbol{\alpha}_t(\mathbf{x}_k), \delta)$ and enter $\mathbb{X}\setminus\mathbb{\hat{X}}(\boldsymbol{\alpha}_t(\mathbf{x}_k), \delta)$ in finite time. Assume by contradiction that the closed-loop system state stays in $\mathbb{\hat{X}}(\boldsymbol{\alpha}_t(\mathbf{x}_k), \delta)$ forever (the system stays in phase 1 forever), which indicates that both $\mathbf{x}_k$ and $\mathbf{x}_{k,f}$ never enter $B_\delta$ by the definition of $\mathbb{\hat{X}}(\boldsymbol{\alpha}_t(\mathbf{x}_k), \delta)$ in \eqref{eq:xhat}. Along the optimal trajectory, since the value function $V$ is non-increasing with respect to the time step k from \eqref{eq:V decrease} and $V$ is also lower bounded by 0 from its definition in \eqref{eq:V}, it follows that $V(\mathbf{x}_k)$ converges as $k\rightarrow \infty$, which implies $\lim_{k\rightarrow\infty}V(\mathbf{x}_k)-V(\mathbf{x}_{k-1})=0$. 
Due to \eqref{eq:V decrease} and the fact that $V(\mathbf{x}_k)$ converges, it follows that
\begin{align}
\label{eq:limit}
    \lim_{k\rightarrow \infty} \boldsymbol{\alpha}_*(\mathbf{x}_{k-1})^\top \mathbf{J}(\mathbf{x}_{k},\mathbf{U}_s(\mathbf{x}_{k}))-\boldsymbol{\alpha}_*(\mathbf{x}_{k-1})^\top \mathbf{J}(\mathbf{x}_{k-1},\mathbf{U}_*(\mathbf{x}_{k-1})) = 0.
\end{align}
Considering \eqref{eq:limit}, \eqref{eq:difference}, and the fact that $\alpha^i_*(\cdot)$ and $L^i(\cdot)$ are non-negative for all $i$, we have
$\lim_{k\rightarrow \infty}\sum_{i=0}^m \alpha_*^i(\mathbf{x}_k)L^i(\mathbf{x}_k,\mathbf{u}^*_k)= 0$.
From \eqref{eq:beta}, \eqref{eq:beta_s} and \eqref{eq:2 beta}, we have $\alpha_*^0(\mathbf{x}_k)\ge\beta_t \ge \beta> 0 \; \forall \mathbf{x}_k$ while the system is still in phase 1.
Since $\alpha_*^0(\mathbf{x}_k)\ge\beta>0$, $\alpha_*^i(\mathbf{x}_k)\geq 0$ and $L^0(\mathbf{x}_k,\cdot)\ge\sigma_1(\|\mathbf{x}_k\|)\; \forall \mathbf{x}_k$, it follows that $\lim_{k\rightarrow\infty}\sum_{i=0}^m \alpha_*^i(\mathbf{x}_k)L^i(\mathbf{x}_k,\mathbf{u}^*_k)= 0$ leads to $\lim_{k\rightarrow\infty}\mathbf{x}_k= 0$. From $\lim_{k\rightarrow\infty}\mathbf{x}_k= 0$, we know $\mathbf{x}_k$ will enter $B_\delta$ in finite time, and we reach a contradiction with the assumption that $\mathbf{x}_k\in\mathbb{\hat{X}}(\boldsymbol{\alpha}_t, \delta) \; \forall k>0$. As a result, we can guarantee that the state of the closed-loop system will leave $\mathbb{\hat{X}}(\boldsymbol{\alpha}_t, \delta)$ and enter $\mathbb{X}\setminus\mathbb{\hat{X}}(\boldsymbol{\alpha}_t, \delta)$ in finite time, i.e., the system enters phase 2 from phase 1 in finite time. Thus, we have (II.i).

(II.ii) According to the backup plan constrained MPC algorithm, when the system is in phase 2, it stays in phase 2 and we have $\boldsymbol{\alpha}_*(x)=[1,0,0,\dots,0]^\top \; \forall \mathbf{x}\in\mathbb{X}$. Thus, the feasibility maximization problem becomes a single objective MPC problem, and only the primary destination is considered in phase 2. With Assumptions \ref{assumption:1} and \ref{assumption 2} satisfied, Assumption III.1 in \cite{lazar2006stabilizing} is satisfied. When the closed-loop system stays in phase 2, we can leverage Theorem III.2 in \cite{lazar2006stabilizing} to prove that the origin of the closed-loop system is asymptotically stable in the Lyapunov sense. Thus, when the system is in phase 2, the closed-loop system state converges to the origin. 

(I) Leveraging Theorem III.2 in \cite{lazar2006stabilizing}, when the closed-loop is in phase 2, we have the following three properties:
\begin{align}
\label{eq:3 properties}
    &V(\mathbf{x}_k) \ge V(\mathbf{x}_{k+1}) \nnum\\
    & V(\mathbf{x}_k) \ge \sigma_l(\|\mathbf{x}_k\|) \nnum\\
    & V(\mathbf{x}_k) \le \sigma_u(\|\mathbf{x}_k\|),
\end{align}
where $\sigma_l$ and $\sigma_u$ are two class $\mathcal{K}$-functions.
Regarding the tuning parameter $\delta$, for any $\epsilon\in(0,\delta)$, we can always find a $\delta'\in (0,\epsilon)$ s.t. $\sigma_u(\delta') < \sigma_l(\epsilon)$. According to Algorithm \ref{al:backup}, for any initial condition $\|\mathbf{x}_0\|< \delta'< \delta$, the system is already in phase 2. Based on \eqref{eq:3 properties}, we achieve $\dots \le V(\mathbf{x}_{k+1})\le V(\mathbf{x}_k)\le\dots\le V(\mathbf{x}_0)\le \sigma_u(\|\mathbf{x}_0\|)<\sigma_u(\delta')<\sigma_l(\epsilon)$, which implies $\mathbf{x}_k\in B_\epsilon \; \forall k\ge0$. For any $\epsilon' > \epsilon$, we can keep picking the same $\delta'$ (which makes the system stay in phase 2) and achieve $\|\mathbf{x}_0\|<\delta' \rightarrow \|\mathbf{x}_k\|<\epsilon < \epsilon'$.
Thus, $\forall \epsilon >0\; \exists \delta' \text{ s.t. } \|\mathbf{x}_0\|< \delta' \rightarrow \|\mathbf{x}_k\| < \epsilon \; \forall k>0$, and we have the Lyapunov stability of the closed-loop system with the backup plan constrained MPC controller defined in Algorithm \ref{al:backup}. 

From the results of (II.i) and (II.ii), it follows that for any initial condition $\mathbf{x}_0\in\mathbb{X}$, we have $\lim_{k\rightarrow\infty}\mathbf{x}_k=0$. With $\lim_{k\rightarrow\infty}\mathbf{x}_k=0 \; \forall \mathbf{x}_0\in\mathbb{X}$ and (I), we prove that the system in the closed loop, with the backup plan constrained MPC algorithm defined in Algorithm \ref{al:backup}, is asymptotically stable with $\mathbb{X}$ as a region of attraction.
\end{proof}

Since we have misaligned objectives in the backup plan constrained control problems, the feedback control $K\mathbf{x}$ used towards the decrease of the cost of the primary destination may result in the increase of the cost toward alternative destnations. Thus, \eqref{eq:assumption} is assumed to hold only for the primary destination. Since $\hat{\mathbf{u}}$ in \eqref{eq:uhat} results in the smallest value of $\max_{\{i,\mathbf{x}\}\in \mathbb{I}\times\mathbb{X}}L^i(\mathbf{x},\mathbf{u})+F^i(f(\mathbf{x},\mathbf{u}))-F^i(\mathbf{x})$ over all feasible control inputs, we add $\hat{\mathbf{u}}$ instead of $K\mathbf{x}$ to $\mathbf{U}_s^i(\mathbf{x}_k)$ in \eqref{eq:u_s} for $i\neq 0$ to achieve a smaller upper bound for the cost increases of the alternative destinations.
The problem preventing us from getting the direct asymptotic stability of the closed-loop system without checking if $\mathbf{x} \in \mathbb{\hat{X}}(\boldsymbol{\alpha}_t(\mathbf{x}_k),\delta)$ or the phase of the system is the condition in \eqref{eq:al 1} when $P >0$. According to the definition of $\beta$ in \eqref{eq:beta} and since $\gamma_i$ and $\delta$ are positive parameters, we have $0<\beta\le1$. Since $\beta=1$ is undesirable as we previously explained in Remark \ref{remark 1}, to ensure the existence of $0<\beta < 1$ for $\beta(L^0(\mathbf{x},K\mathbf{x})+F^0(f(\mathbf{x},K\mathbf{x}))-F^0(\mathbf{x})) \le-(1-\beta) P$, \eqref{eq:al 1} must only hold for inputs outside a ball $B_\delta$ around the origin. Since the state we plug into \eqref{eq:al 1} is actually $\mathbf{x}_{k,f}$, and because \eqref{eq:al 1} is necessary for the non-increase of the value function $V$, we can only achieve the non-increase of value function $V$ for $\mathbf{x}_k\notin B_\delta \text{ s.t. } \mathbf{x}_{k,f}\notin B_\delta$, i.e., $\mathbf{x}_k\in\mathbb{\hat{X}}(\boldsymbol{\alpha}_t(\mathbf{x}_k),\delta)$ and the system is in phase 1. After $\mathbf{x}_k$ leaves $\mathbb{\hat{X}}(\boldsymbol{\alpha}_t(\mathbf{x}_k),\delta)$ and when the system is in phase 2, to ensure that the value function does not increase, we change the design of the weight vector to $\boldsymbol{\alpha}_*(x)=[1,0,0,\dots,0]^\top \;\forall \mathbf{x}$. We can thus guarantee the non-increase of the value function $V$ along the entire system trajectory.
In addition to the backup plan constrained control problems, this proof can be generalized to all multi-objective MPC problems with misaligned objectives.

\section{Multi-horizon Multi-objective Model Predictive Path Integral Control (3M)}\label{sec:3M}

The current section proposes the 3M solver, which is used in the backup plan constrained MPC algorithm to solve the feasibility maximization problem~\eqref{eq:mmpc} efficiently given a design of the weight vector.
The proposed solver is based on MPPI, which is sampling-based and suitable for parallel computation.

\subsection{Model Predictive Path Integral Control (MPPI)}\label{sec:MPPI}

Through parallel computation, MPPI solves the stochastic optimal control problem by sampling the system trajectories  \citep{williams2016aggressive,williams2017model, williams2018information}.
Due to the sampling nature of MPPI, it is not necessary to obtain the derivatives of either the dynamics or the cost function of the system. Thus, we are able to solve the optimal control problem for nonlinear dynamics and non-smooth/non-differentiable cost functions without approximations. MPPI is capable of real-time application even for relatively large dimensions of the state space (e.g., there are 48 state variables for the 3-quadrotor control example in \citep{williams2017model}) using GPUs. As a result, MPPI has been used widely to handle real-time control problems due to its computational efficiency.

For dynamical system~\eqref{eq:sysmodel}, consider a noise-corrupted input:
$
    \mathbf{x}_{k+1}= f(\mathbf{x}_k,\mathbf{u}_{k}+\epsilon_k),
$
where $\epsilon_k \in \RR^{n_u}$ is an independent and identically distributed (i.i.d.) Gaussian noise, i.e., $\epsilon_k \sim \mathcal{N}(0, \Sigma)$ with the known co-variance matrix $\Sigma$.
The optimization problem aims to determine an input sequence $\mathbf{U}^0= [\mathbf{u}_0^\top, \mathbf{u}_1^\top, \cdots, \mathbf{u}_{N-1}^\top]^\top$ with finite time horizon N that minimizes the expected cost
$
        \mathbf{U}^* = \argmin_{\mathbf{U}^0} \mathbb{E} [S(\tau)],
$
where $\tau = \{\mathbf{x}_0, \mathbf{u}_0, \mathbf{x}_1, \mathbf{u}_1,\cdots, \mathbf{u}_{N-1}, \mathbf{x}_N\}$ is the state-input set. The cost of a trajectory can be calculated by:
\begin{align}\label{eq:costMPPI}
    S(\tau) = \phi(\mathbf{x}_N) + \sum_{t=0}^{N-1} \left( c(\mathbf{x}_t) + \lambda \mathbf{u}_t^\top\Sigma^{-1}\epsilon_t  \right),
\end{align}
where $\phi(\mathbf{x}_N)$ is the terminal cost, $c(\mathbf{x}_t)$ is the state-dependent running cost, and $\lambda$ is a tuning parameter. The random noise $\epsilon_k$ in the input for $k=0,\cdots, N-1$ affects the state-input set $\tau$ and thus influences the cost function \eqref{eq:costMPPI}. We sample $\epsilon_k$ from the distribution $\mathcal{N}(0, \Sigma)$, and construct $M$ noise-corrupted trajectories $\boldsymbol\epsilon^q = [(\epsilon_0^q)^\top,(\epsilon_1^q)^\top,\cdots,(\epsilon_{N-1}^q)^\top]$ for $q=1,\cdots,M$, which are then evaluated with the cost function $S$ from \eqref{eq:costMPPI}.
Furthermore, the iterative update law from \citep{williams2017model} is leveraged by MPPI to calculate the current optimal input sequence $\mathbf{U}_k^*$ based on the previous optimal input sequence $\mathbf{U}_{k-1}^*$ as follows: $\mathbf{U}_k^* = [\mathbf{u}_{1,k-1}^\top,\mathbf{u}_{2,k-1}^\top,\cdots,\mathbf{u}_{N-1,k-1}^\top,{u'}^\top]^\top +\sum_{q=1}^Mw^q{\boldsymbol\epsilon}^q,$
where
$
    w^q=\frac{exp(-\frac{1}{\lambda }S(\tau^q))}{\sum_{q=1}^M exp(-\frac{1}{\lambda}S(\tau^q))}
$
with the temperature parameter of the Gibbs distribution $\lambda$, and the set of state-input for the $q^{th}$ noise trajectory $\tau^q$. Regarding $\mathbf{U}_{k}^*$, $\mathbf{u}_{i,k-1}$ is the $i^{th}$ input of $\mathbf{U}_{k-1}^*$ and $u'$ is a custom control input.

\begin{algorithm}[t]
\caption{3M Solver}\label{al:3}
\textbf{Input to the 3M Solver:}\\
$\mathbf{x}_k$: Current state; \\
$\alpha(\mathbf{x}_k)$: Design of the weight vector;\\
$\mathbf{U}_{s}(\mathbf{x}_k)$: Predefined control input sequence;\\
\textbf{Choose tuning parameters:}\\
$M$: Number of sample trajectories;\\
$\Sigma$: Co-variance of the noise $\epsilon_k$;\\
$\lambda$: Temperature parameter of the Gibbs distribution;\\

\begin{algorithmic}[1]
\STATE Sample $M$ trajectories of noise $\boldsymbol\epsilon^q$ as in~\eqref{eq:Ksamplenoise};
\STATE Simulate $\mathbf{U}_s(\mathbf{x}_k)+\boldsymbol\epsilon^q$ on the system~\eqref{eq:sysmodel} to get
$\mathbf{X}^q(\mathbf{x}_k)$ for $q=1,\cdots,M$;
\STATE Evaluate $\alpha(\mathbf{x}_k)^\top \mathbf{J}^q$ for $q=1,\cdots,M$, where $\mathbf{J}^q = \mathbf{J}(\mathbf{x}_k,\mathbf{U}_s(\mathbf{x}_k)+\boldsymbol\epsilon^q)$;
\STATE Calculate estimated output control $\mathbf{U}(\mathbf{x}_k,\alpha(\mathbf{x}_k))$ in~\eqref{eq:3m_input} using the costs of the $M$ trajectories,
$(\alpha(\mathbf{x}_k)^\top\mathbf{J}^1,\cdots,\alpha(\mathbf{x}_k)^\top\mathbf{J}^M)$ and the weights $\boldsymbol w^q$.
\end{algorithmic}
\textbf{Return: $\mathbf{U}(\mathbf{x}_k,\alpha(\mathbf{x}_k))$}
\end{algorithm}

\subsection{Multi-horizon Multi-objective Model Predictive Path Integral Control (3M) Solver}

In this section, we propose the multi-horizon multi-objective model predictive path integral control (3M) solver through the application of MPPI control to the feasibility maximization problem \eqref{eq:mmpc} given a design of the weight vector. The proposed 3M solver is summarized in Algorithm \ref{al:3} and explained below.

The input to the 3M solver includes the current state $\mathbf{x}_k$, the design of weight vector $\alpha(\mathbf{x}_k)$ from the backup plan constrained MPC algorithm, and $\mathbf{U}_s(\mathbf{x}_k)$, a control input sequence generated from the previous optimal control input sequence $\mathbf{U}_*(\mathbf{x}_{k-1})$.

We first sample $M$ trajectories of noises as follows (line 1):
\begin{align} \label{eq:Ksamplenoise}
    \boldsymbol\epsilon^q \in \RR^{(N+\frac{N(N-1)}{2}m)n_u}, \quad
    \epsilon^q(i) \sim {\mathcal{N}}(0,\Sigma)
\end{align}
for $q=1,\cdots,M$ and $i=1,\cdots,N+\frac{N(N-1)}{2}m$, where $\epsilon^q(i) \in \RR^{n_u}$ is the $i^{th}$ noise vector of $\boldsymbol\epsilon^q$. Then, we construct noise disturbed input $\mathbf{U}_s(\mathbf{x}_k)+\boldsymbol\epsilon^q$ for each sampled trajectory with noise, and evaluate the corresponding cost as (lines 2 and 3):
\begin{align*}
    \alpha(\mathbf{x}_k)^\top \mathbf{J}^q = \alpha(\mathbf{x}_k)^\top \mathbf{J}(\mathbf{x}_k,\mathbf{U}_s(\mathbf{x}_k)+\boldsymbol\epsilon^q) \quad \text{for $q=1,\cdots,M$},
\end{align*}
and we denote $\mathbf{X}^q(\mathbf{x}_k)$ as the simulated state trajectory with the noise disturbed input $\mathbf{U}_s(\mathbf{x}_k)+\boldsymbol\epsilon^q$ on the system~\eqref{eq:sysmodel}. Given the input weight vector $\alpha(\mathbf{x}_k)$ and based on the costs of $M$ noise-corrupted trajectories $(\alpha(\mathbf{x}_k)^\top\mathbf{J}^1,\cdots,\alpha(\mathbf{x}_k)^\top\mathbf{J}^M)$, we calculate the weights as $
    \boldsymbol w^q=\frac{exp(-\frac{1}{\lambda}\alpha(\mathbf{x}_k)^\top\mathbf{J}^q)
}{\sum_{q=1}^M exp(-\frac{1}{\lambda}\alpha(\mathbf{x}_k)^\top\mathbf{J}^q)} \in \RR_{>0}
$
for $q=1,\cdots,M$, and finally obtain the optimal control input sequence as (line 4):
\begin{align}\label{eq:3m_input}
    \mathbf{U}(\mathbf{x}_k,\alpha(\mathbf{x}_k))=\mathbf{U}_s(\mathbf{x}_k)
+\sum_{q=1}^M \boldsymbol w^q{\boldsymbol\epsilon}^q.
\end{align}

\section{Illustrative Examples}\label{sec:sim}

In this section, we tested the proposed algorithm in simulations to address the motion planning problem of the unmanned aerial vehicle (UAV) with two different dynamical systems. In both experiments, the control objective for the UAV is to arrive at the primary destination $[0,0]$  from the initial position $[5,9]$ in the $2$-D plane. The UAV simulations here can be a real-world representation of an urban drone delivery system or an air taxi transportation system, where the aerial vehicle flies $300$ ft above ground and delivers a package or passenger to the primary destination. Safe rooftops or open areas with no traffic or people can be selected as alternative destinations for an emergency landing in case of primary destination abortion. For computational efficiency, we use GPUs for parallel computation and also discuss applicable control frequency in the real world in terms of calculation time for solving the feasibility maximization \eqref{eq:mmpc}.

\subsection{Simulation Results} \label{sec:simulation}
We first evaluate the performance of the proposed controller on the following dynamical model:
\begin{align}
\label{eq:model2}
    \mathbf{x}_{k+1}=
    \left[
    \begin{array}{cccc}
     1 & 0 & 0.1 & 0   \\
     0 & 1 &  0 & 0.1 \\
     0 & 0 &  1 & 0   \\
     0 & 0 &  0 & 1
    \end{array}
    \right] \mathbf{x}_k +
    \left[
    \begin{array}{cccc}
     0 & 0    \\
     0 & 0    \\
     1 & 0    \\
     0 & 1
    \end{array}
    \right]\mathbf{u}_{k} , 
\end{align}
where $\mathbf{x}_k \in \RR^4$ represents the horizontal coordinate, vertical coordinate, horizontal velocity, and vertical velocity, and the input $\mathbf{u}_{k} \in \RR^2$ consists of horizontal and vertical accelerations. The initial condition is $\mathbf{x}_0 = [5,9,0,0]^\top$, and the primary destination is $\mathbf{p}^0 = [0,0,0,0]^\top$. To satisfy  Assumption \ref{assumption 2} and constraints \eqref{eq:beta} and \eqref{eq:al 1}, one needs to set the domains $\mathbb{X}$ and $\mathbb{U}$. The state constraint is set as $\mathbf{x}_k=[x_1,x_2,x_3,x_4]^\top$, where $x_1 \in [-2,10]$, $x_2\in[-2,10]$, $x_3\in[-10,2]$ and $x_4 \in [-10,2]$, and the input constraint is $\mathbf{u}_{k}=[u_1,u_2]^\top$, where $u_1\in[-10,2]$ and $u_2\in[-10,2]$. Euclidean distance has been used for the distance metric $d$ in~\eqref{eq:MissionCompletetion}. The cost function in~\eqref{eq:costSmallJ} is constructed as $L^i(\mathbf{x},\mathbf{u})= (x-\mathbf{p}^i)^\top Q^1 (x-\mathbf{p}^i)+u^\top R u$, $F^i(\mathbf{x}) = (x-\mathbf{p}^i)^\top Q^2 (x-\mathbf{p}^i)$ with $Q^1=10^{-5}*\mathbb{I}$, $Q^2=0.1*\mathbb{I}$ and $R=0.1*\mathbb{I}$, where $\mathbb{I}$ is the identity matrix with a proper dimension. The matrix $K$ is chosen such that the pole of the system $A+BK$ is placed at $[0.95;0.9;0.95;0.9]$, and also assumption \eqref{eq:assumption} is satisfied with $\delta=2$. The MPPI control parameters are $\Sigma={\mathbb{I}}$, $\lambda = 1$, and $M=10000$. The control horizon is set as $N=10$. Two simulations have been conducted with two different alternative destinations ($m=2$). For the first simulation, alternative destinations are given by $\mathbf{p}^1=[4,9,0,0]^\top$ and $\mathbf{p}^2 = [1,4,0,0]^\top$ with calculated $P=0.56$. 
In order to satisfy \eqref{eq:al 1}, we need $\beta > 0.65$. For the second simulation, $\mathbf{p}^1=[4,6,0,0]^\top$ and $\mathbf{p}^2 = [3,1,0,0]^\top$ with calculated $P=0.40$, and we also need $\beta > 0.56$ to satisfy \eqref{eq:al 1}. The tuning parameters $\gamma_i$ and $\mu$ of the backup plan constrained MPC algorithm are chosen differently for each simulation to satisfy \eqref{eq:beta} and \eqref{eq:al 1}, and those parameters are listed in each figure. The simulation results are presented in Fig. \ref{fig:sim2}, with the top two figures showing the generated trajectory and time history of the weight vector of the first simulation and the bottom two figures showing the results of the second simulation. For the trajectory plot, the primary destination is marked by a black dot, alternative destination 1 by a red dot, alternative destination 2 by a green dot, and the initial position by a blue dot. 
Based on Fig. \ref{fig:sim2}, the UAV with the backup plan constrained MPC algorithm makes a detour to the primary destination in both cases, flying near alternative destinations. The detour trajectory is safer in the backup plan sense, providing a shorter path toward one of the alternative destinations when an emergency landing is in need. The value of $\boldsymbol{\alpha}_*(\mathbf{x}_k)$ determined by the backup plan constrained MPC algorithm is also presented in Fig. \ref{fig:sim2}. We see that $\alpha_*^i(\mathbf{x}_k)$ changes dynamically based on the current state $\mathbf{x}_k$, and $\alpha_*^0(\mathbf{x}_k)$ does not decrease as the system state gets closer to the primary destination. After $\mathbf{x}_k$ enters $\mathbb{\hat{X}}(\boldsymbol{\alpha}_t(\mathbf{x}_k), \delta)$ and the system is in phase 2, we see that $\boldsymbol{\alpha}_*(\mathbf{x}_k)$ becomes $[1,0,0]^T$ as expected. In addition, along the system trajectory, $\boldsymbol{\alpha}_*(\mathbf{x}_k)$ stays constant for a period of time before entering phase 2, which is caused by our design of $\boldsymbol{\alpha}_t(\mathbf{x}_k)$ in \eqref{eq:alpha_t} as we may have $\boldsymbol{\alpha}_t(\mathbf{x}_k)=\boldsymbol{\alpha}_*(\mathbf{x}_{k-1})$ during these times.
\begin{figure}[!t]
 \centering
	\includegraphics[width=0.8\linewidth]{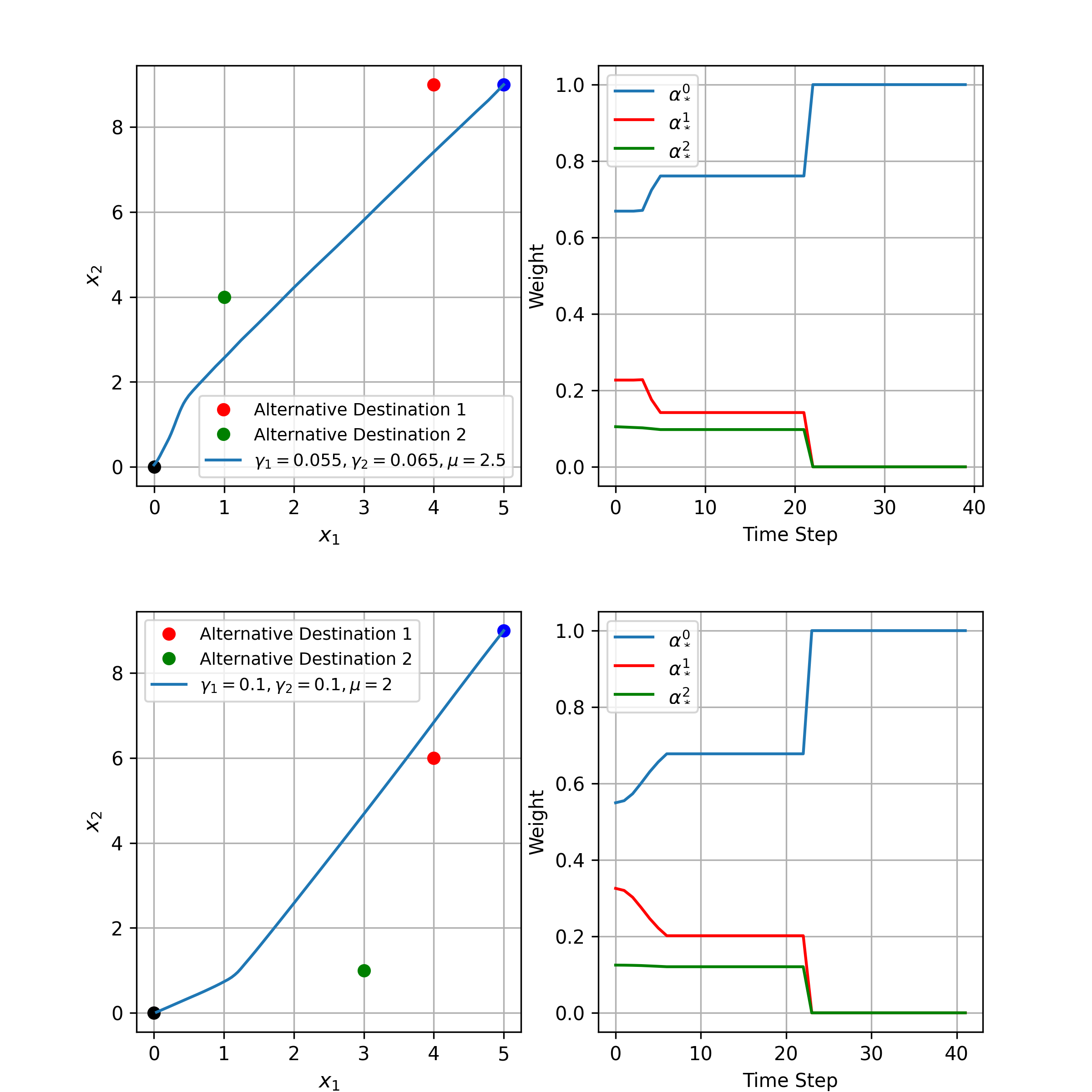}
\caption{UAV simulation results for system \eqref{eq:model2} with two sets of alternative destinations.}
\label{fig:sim2}
\end{figure}

To evaluate the performance of the proposed backup plan constrained MPC algorithm in terms of increasing the possibility of an emergency landing at an alternative destination, we include experiments using flight scenarios in which vehicles must navigate to the nearest destination at a random time step during the operation en route to the primary destination. The purpose of this event is to simulate an emergency landing due to unexpected failure. The random failure time was uniformly sampled from the finite set, i.e., $\{1, 2, 3, \dots, 20\}$. For comparison, we also used a baseline method that only aims to arrive at the primary destination before failure. In contrast, our proposed framework considers both the primary and the alternative destinations before failure. 
The distance to the closest destination at the time of randomly sampled failure averaged over 50 sample flights, is presented in Table~\ref{tab:random_test3} and Table~\ref{tab:random_test4}. As shown in the third column (labeled as \emph{Distance at fail.}) of the tables, the proposed framework tends to have a relatively shorter distance to the closest destination at the onset of the failure than the baseline. In addition to the distance to the closest destination at the time of failure, we also reported the energy, i.e., the sum of squares of the control input, used to land on the closest destination (labeled as \emph{Energy used after fail.}) in the tables. The results show that the proposed method uses much less energy to land at the closest destination after failure compared with the baseline, saving 32\% of energy for the first simulation setup and 24\% for the second simulation setup. Furthermore, the proposed control framework tends to use less energy in the entire flight, as shown in the column labeled as \emph{Total energy used} of the tables. Finally, we introduced a metric relevant to safety margin, i.e., $\text{Margin} = \text{Average remaining energy at failure}/\text{Average energy used after failure}$. If this metric is less than one, it means that the vehicle could not reach the closest destination. As shown in the last column of the tables, with the total energy set as 8, the proposed algorithm has significantly greater margins than the baseline, indicating a larger possibility for an emergency landing.
\begin{table}[t]
    \centering
    \caption{Random failure test results for first simulation setup of system \eqref{eq:model2}}
    \begin{tabular}{lccccc}
    \hline
    \hline
     Method &Time of failure& Distance at fail.& Energy used after fail. & Total energy used & Margin at fail.\\
     & (avg$\pm$stdev)&(avg$\pm$stdev)  &(avg$\pm$stdev)          & (avg$\pm$stdev)   &    (total energy = 8)      \\
    \hline
    Proposed&  $10.10\pm5.1$  &$1.42\pm0.7$ &$1.15\pm0.6$ &$1.95\pm0.8$ &8.2\\ 
    Baseline&  $8.49\pm5.8$  &$1.50\pm0.7$ &$1.69\pm1.0$ &$3.82\pm1.7$ &5.8\\
    \hline
    \hline
    \end{tabular}
    \label{tab:random_test3}
\end{table}

\begin{table}[t]
    \centering
    \caption{Random failure test results for second simulation setup of system \eqref{eq:model2}}
    \begin{tabular}{lccccc}
    \hline
    \hline
     Method &Time of failure& Distance at fail.& Energy used after fail. & Total energy used & Margin at fail.\\
     & (avg$\pm$stdev)&(avg$\pm$stdev)  &(avg$\pm$stdev)          & (avg$\pm$stdev)   &      (total energy = 8)      \\
    \hline
    Proposed&  $11.75\pm5.3$  &$1.75\pm0.9$ &$1.25\pm0.5$ &$2.21\pm0.8$ &7.0\\ 
    Baseline&  $10.49\pm5.6$  &$1.80\pm0.8$ &$1.65\pm0.8$ &$4.18\pm1.0$ &4.8\\
    \hline
    \hline
    \end{tabular}

    \label{tab:random_test4}
\end{table}
Based on the results above, it can be concluded that the proposed backup plan-constrained MPC algorithm can significantly increase safety in terms of backup plan safety for this UAV path planning problem, increasing the chance to safely land at an alternative destination during the operation toward the primary destination. While the performance of the proposed algorithm is excellent when applied to system \eqref{eq:model2}, the generated trajectory in Fig. \ref{fig:sim2} does not have an optimal visual effect for illustration purpose. To help the readers better understand the impact and benefit of the proposed algorithm, we evaluate the performance of the proposed algorithm on a second dynamical model: 
\begin{align}
\label{eq:model1}
    &\mathbf{x}_{k+1}=
    \left[
    \begin{array}{cccc}
     1 & 0 \\
     0 & 1 \\
    \end{array}
    \right] \mathbf{x}_k +
    \left[
    \begin{array}{cccc}
     1 & 0    \\
     0 & 1
    \end{array}
    \right]\mathbf{u}_{k} ,
\end{align}
where $\mathbf{x}_k \in \RR^2$ represents the horizontal coordinate and the vertical coordinate, and the control input $\mathbf{u}_{k} \in \RR^2$ comprises the horizontal and vertical speeds. 
For system \eqref{eq:model1}, we consider the state constraint set as $\mathbf{x}_k = [x_1,x_2]^\top$, where $x_1\in [-2,10]$ and $x_2\in[-2,10]$, and the input constraint set as $\mathbf{u}_{k}=[u_1,u_2]^\top$, where $u_1\in [-10,2]$ and $u_2 \in[-10,2]$.
The initial condition is $\mathbf{x}_0 = [5,9]^\top$ and the primary destination is $\mathbf{p}^0 = [0,0]^\top$. We used the same distance metric
$d$, cost functions, and MPPI control parameters in this setup, and the control horizon is set as $N$ = 5. The matrix $K$ is chosen such that the pole of the system $A+BK$ is placed at $[0.9;0.9]$, and also \eqref{eq:assumption} is satisfied with $\delta=3$. 
Two simulations have been conducted with two sets of alternative destinations ($m=2$). For the first simulation, alternative destinations are given by $\mathbf{p}^1=[3,9]^\top$ and $\mathbf{p}^2 = [1,5]^\top$ with calculated $P=0.0099$. In order to satisfy \eqref{eq:al 1}, we need $\beta > 0.0061$. For the second simulation, alternative destinations are $\mathbf{p}^1=[4,6]^\top$ and $\mathbf{p}^2 = [3,1]^\top$ with calculated $P=0.0044$. In order to satisfy \eqref{eq:al 1}, we need $\beta > 0.0027$. The simulation results are presented in Fig. \ref{fig:sim1}.
\begin{figure}[!b]
 \centering
\includegraphics[width=0.8\linewidth]{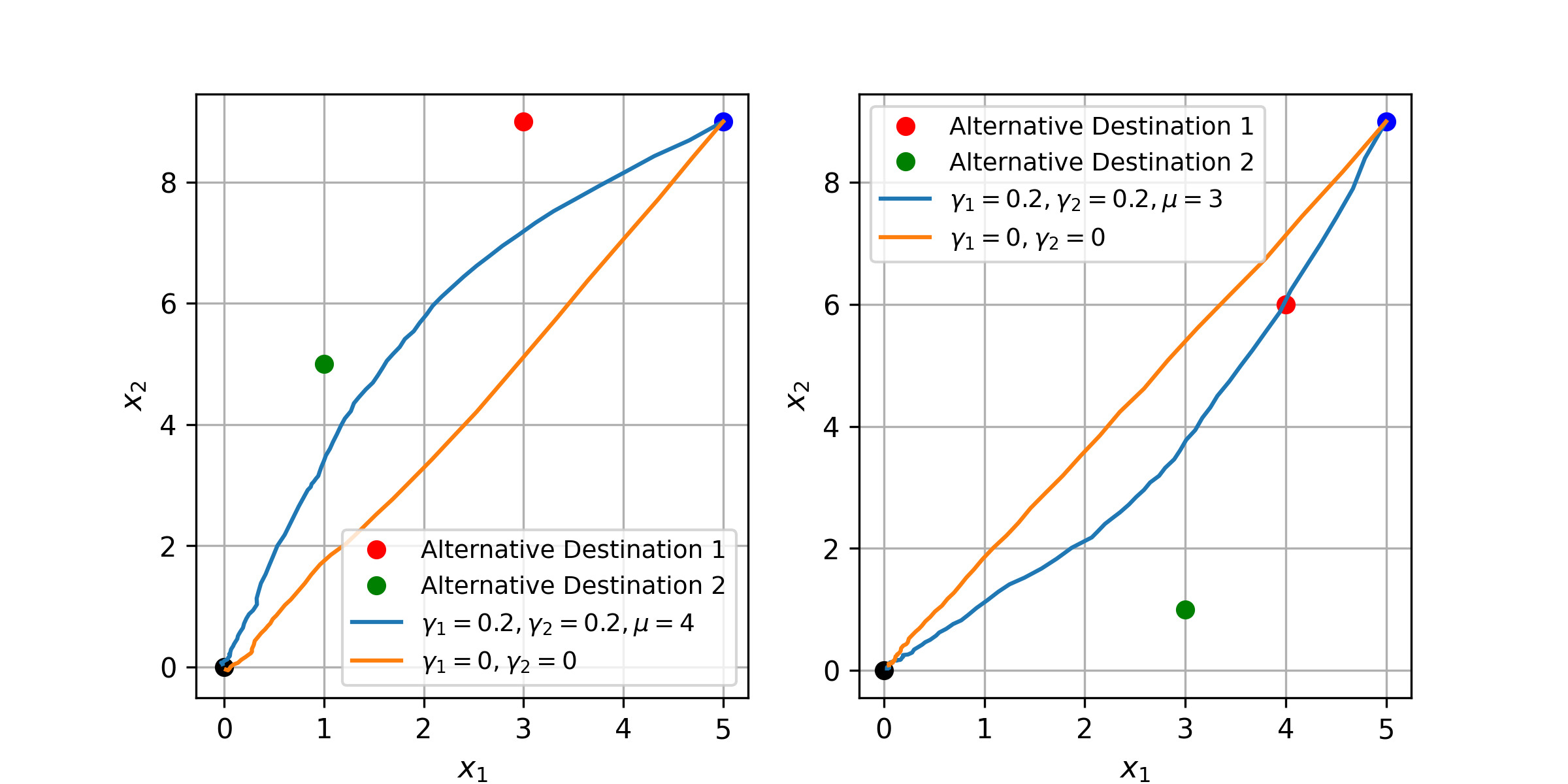}
\caption{UAV simulation results for system \eqref{eq:model1} with two sets of alternative destinations.}
\label{fig:sim1}
\end{figure}
The left figure presents the simulation results with $\mathbf{p}^1=[3,9]^\top$ and $\mathbf{p}^2 = [1,5]^\top$. The right one shows the simulation results with $\mathbf{p}^1=[4,6]^\top$ and $\mathbf{p}^2 = [3,1]^\top$. The blue line is the executed state trajectory of the UAV following the backup plan constrained MPC algorithm, and the orange line is the baseline trajectory using the regular MPPI in which we consider the primary objective only.

The UAV with regular MPPI control ($\gamma = 0$ with primary destination consideration only) flies to the primary destination almost directly, as expected. The UAV with the backup plan constrained MPC algorithm makes a significant detour to the primary destination in both cases, flying near alternative destinations. 
Regarding the simple system \eqref{eq:model1}, as we have control input towards both elements of the state, we do not have a strict constraint for $\beta$, a lower bound of $\alpha_*^0$, and thus the weights we assigned for the alternative destinations can be large. As a result, the detour trajectory is much closer to the alternative destinations. However, for the double integrator system \eqref{eq:model2}, in order to satisfy the assumptions of Theorem \ref{theorem3}, $\beta$ needs to be much larger and thus the weights assigned for alternative destinations are smaller. Accordingly, $\alpha_*^0(\mathbf{x}_k)$ is always large, and the trajectory detours are not considerably closer to the alternative destinations. In terms of visually illustrating the benefits and impact of the proposed backup plan-constrained MPC algorithm, Fig. \ref{fig:sim1} from system \eqref{eq:model1} emerges as a more effective choice when compared to Fig. \ref{fig:sim2} from system \eqref{eq:model2}. 
While system \eqref{eq:model1} features velocities as the control inputs, it can be utilized for the high-level control loop in the cascaded (position) control of the UAV. The control output of this loop, specifically velocity, will be transmitted to the low-level controller to generate the corresponding forces or accelerations necessary for tracking this velocity. Moreover, additional constraints can be introduced to the feasibility maximization \eqref{eq:mmpc} to ensure that the complex system can effectively track the optimized control input velocity obtained from the simple system. The single integration system models are used by researchers on path planning of UAVs, and the cascaded control scheme, which involves mapping desired velocities to accelerations or forces, e.g., individual rotor thrusts for quadrotors, has already been thoroughly studied and validated in previous works \cite{kaufmann2022benchmark,giusti2015machine,loquercio2018dronet}.

We also tested the proposed algorithm on system \eqref{eq:model1} with experiments including emergency landings due to unexpected failure, and the results are shown in Table. \ref{tab:random_test1} and Table. \ref{tab:random_test2}. It can be concluded that the proposed algorithm achieves better performance compared with the baseline method in terms of a smaller distance to the closest destination at the time of failure, a smaller energy used for an emergency landing, and a larger margin at failure. Thus, the proposed backup plan constrained MPC algorithm achieves good performance on both systems.

\begin{table}[!h]
\caption{Random failure test results for first simulation of system \eqref{eq:model1}}
    \centering
    \begin{tabular}{lccccc}
    \hline
    \hline
     Method &Time of failure& Distance at fail.& Energy used after fail. & Total energy used & Margin at fail.\\
     & (avg$\pm$stdev)&(avg$\pm$stdev)  &(avg$\pm$stdev)          & (avg$\pm$stdev)   &    (total energy = 5)      \\   
    \hline
    Proposed&  $5.0\pm 2.5$ & $1.6\pm0.3$ & $0.221\pm0.118$ & $0.885\pm0.276$ &19.6\\ 
    Baseline&  $5.3\pm2.5$  & $2.0\pm0.3$ & $0.314\pm0.130$ & $3.832\pm0.945$ &3.7\\
    \hline
    \hline
    \end{tabular}
    
    \label{tab:random_test1}
\end{table}

\begin{table}[!h]
\caption{Random failure test results for second simulation of system \eqref{eq:model1}}
    \centering
    \begin{tabular}{lccccc}
    \hline
    \hline
     Method&Time of failure& Distance at fail.& Energy used after fail. & Total energy used & Margin at fail.\\
     & (avg$\pm$stdev)&(avg$\pm$stdev)  &(avg$\pm$stdev)          & (avg$\pm$stdev)   &   (total energy = 5)       \\
    \hline
    Proposed&  $4.4\pm2.5$  &$1.4\pm0.8$ &$0.241\pm0.143$ &$1.307\pm0.535$ &16.2\\ 
    Baseline&  $5.4\pm2.5$  &$1.9\pm0.6$ &$0.272\pm0.098$ &$3.834\pm1.032$ &3.7\\
    \hline
    \hline
    \end{tabular}
    \label{tab:random_test2}
\end{table}

\subsection{Relationship between the Tuning Parameters and Computational Time and Cost}
\begin{figure}[!t]
 \centering
  \includegraphics[width=1\linewidth]{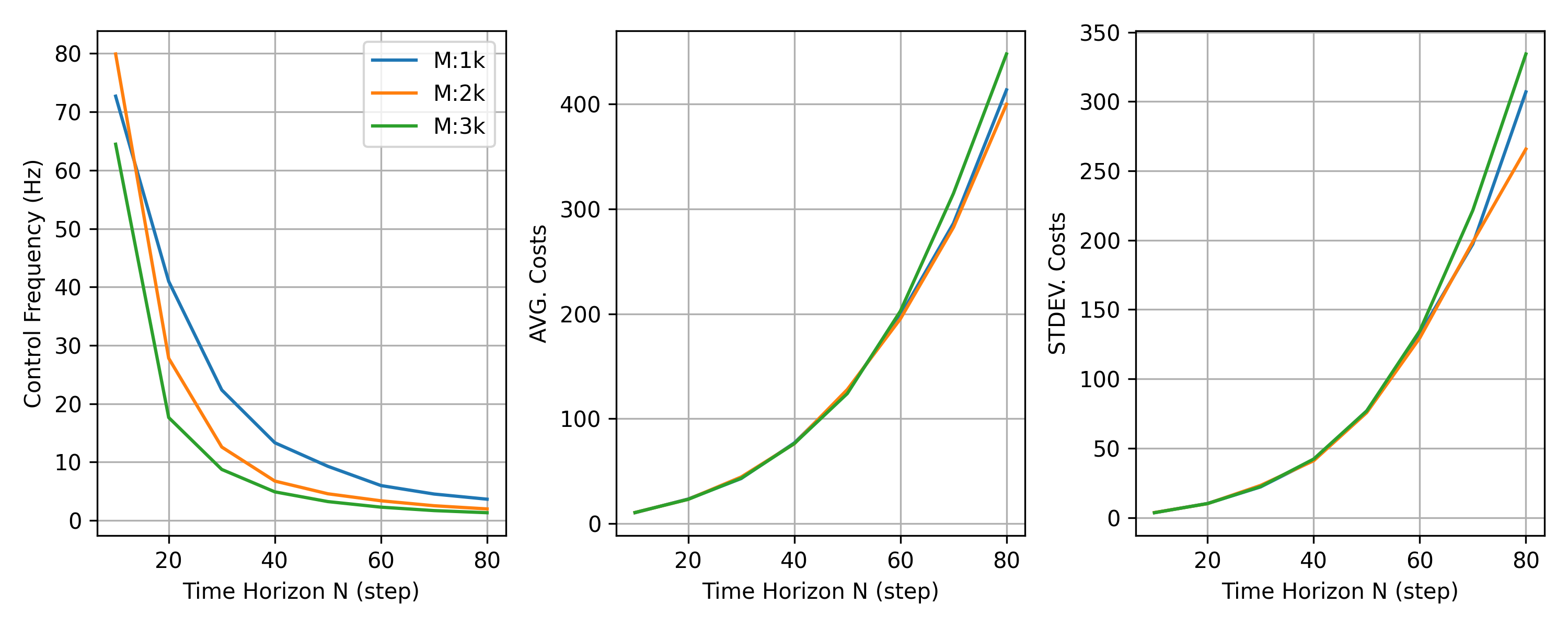}
  \caption{Control frequency and cost over different sample size and prediction horizon.}
\label{fig:cost}
\end{figure}
In this section, we investigate the impact of varying the time horizon $N$ or sample size $M$ on the computational time required to solve the feasibility maximization problem \eqref{eq:mmpc} using the sampling-based 3M solver. Solutions from the sampling-based algorithm become closer to the optimal when the number of samples increases. However, increasing the number of samples trades off computational complexity. Additionally, we examine the effect of different time horizons and sample sizes on both the average and standard deviation of the cost samples $\boldsymbol{\alpha}_*(\mathbf{x})\mathbf{J}(\mathbf{x},\mathbf{U}_*(\mathbf{x}))$. The computational time is inversely correlated with the control frequency of the autonomous vehicle, while the average cost serves as a performance indicator for the proposed algorithm. The experiments have the same setup as the first simulation experiment of system \eqref{eq:model2} discussed in Section \ref{sec:simulation}, and the results are shown in Fig.~\ref{fig:cost}. The figure shows a preliminary result on the relationship between the computational complexity and cost (optimality) with respect to different decision variables in the sampling-based algorithm. 

The first sub-figure in Fig.~\ref{fig:cost} presents the control frequency in relation to the prediction horizon $N$.  Our simulation was executed on a desktop computer equipped with an AMD Ryzen 5 3600 CPU, 16GB RAM, and NVIDIA GeForce RTX 2070 GPU, utilizing the proposed 3M solver. As the prediction horizon $N$ increases, so does the computational complexity of the MPPI part in the order of $N^2$. With a sample size of 1000 and a time horizon of 50, we can run the algorithm every 0.1 seconds which can be considered suitable for real-time implementation.
The second and third sub-figures in Fig.~\ref{fig:cost} depict the impact of time horizon $N$ on the average cost and the standard deviation of the cost samples $\boldsymbol{\alpha}_*(\mathbf{x})\mathbf{J}(\mathbf{x},\mathbf{U}_*(\mathbf{x}))$ with different sample sizes. Increasing $N$ is expected to increase the cost since running costs accumulate with an extended time horizon as defined in \eqref{eq:costSmallJ}, which is aligned with the results shown in Fig. \ref{fig:cost}. Also, the increasing variance in N can be explained by
an accumulation of variance due to the addition of random variables. Therefore, given a desired control frequency, we recommend testing the proposed algorithm with a relatively small sample size and a small time horizon for the 3M solver initially. Then, we can keep increasing the time horizon as long as the desired control frequency can be maintained. If the simulation results are still not good enough, increasing the sample size while adjusting the time horizon to satisfy the control frequency requirement is recommended.

\section{Conclusion}
This paper presents a new safety concept, called backup plan safety, to enhance the path planning of autonomous systems. This concept was inspired by the need to consider alternative missions when the primary mission is deemed infeasible. To ensure safety against mission uncertainties, we formulated a feasibility maximization problem based on multi-objective MPC and incorporated a weight vector to balance the multiple costs and multi-horizon inputs to evaluate the costs of the alternative missions. The proposed backup plan constrained MPC algorithm designs the weight vector that ensures the asymptotic stability of the closed-loop system by making the cost reduction of the primary mission dominate the cost increases of the alternative destinations, thereby achieving a non-increasing overall cost. To address the feasibility maximization problem with multi-horizon inputs efficiently, the backup plan constrained MPC algorithm leverages the proposed multi-horizon multi-objective model predictive path integral control (3M) solver, which uses a sampling-based method to enable efficient parallel computing. Preliminary numerical simulations of UAV path planning illustrate this new concept and the effectiveness of the proposed algorithm. 

\section{Acknowledge}
This work was supported in part by AFOSR under Grant FA9550-21-1-0411, NASA under Grants 80NSSC20M0229 and 80NSSC22M0070, NSF under Grants CNS-1932529 and IIS-2133656.

\bibliography{main}

\end{document}